%% file: arXiv2.tex
\documentclass[a4paper,allowtoday,onecolumn,accepted=2026-04-12]{quantumarticle} 
\pdfoutput=1 

\usepackage{dsfont} 	
\usepackage{microtype} 	
\usepackage[numbers,sort&compress]{natbib} 	
\usepackage[utf8]{inputenc} 	
\usepackage[T1]{fontenc} 	
\usepackage[british]{babel} 	
\usepackage[dvipsnames,x11names]{xcolor} 	
\usepackage{amsmath,amssymb,amsthm,bm,amsfonts,bbm} 	
\usepackage{mathtools} 	
\usepackage{physics} 	
\usepackage[unicode,colorlinks=true,citecolor=OliveGreen,linkcolor=Purple,urlcolor=Magenta]{hyperref} 	
\usepackage{thmtools}
\usepackage{thm-restate}
\usepackage{caption} 
\usepackage{subcaption}

\input{0_MTQ_macros}  

\begin{document}
\title{Multicopy quantum state teleportation with application to storage and retrieval of quantum programs}
\author{Frédéric Grosshans}
\orcid{0000-0001-8170-9668}
\affiliation{Sorbonne Universit\'{e}, CNRS, LIP6, F-75005 Paris, France}
\author{Michał Horodecki}
\orcid{0000-0002-0446-3059}
\affiliation{International Centre for Theory of Quantum Technologies, University of Gdańsk, Jana Bażyńskiego 1A, 80-309 Gdańsk, Poland}
\author{Mio Murao}
\orcid{0000-0001-7861-1774}
\affiliation{Department of Physics, Graduate School of Science, The University of Tokyo, 7-3-1 Hongo, Bunkyo-ku, Tokyo 113-0033, Japan}
\affiliation{Trans-scale Quantum Science Institute, The University of Tokyo, Bunkyo-ku, Tokyo 113-0033, Japan}
\author{Tomasz Młynik}
\orcid{0000-0003-4182-907X}
\affiliation{Institute of Theoretical Physics and Astrophysics, Faculty of Mathematics, Physics and Informatics, University of Gda\'{n}sk, Wita Stwosza 57, 80-308 Gda\'{n}sk, Poland}
\author{Marco Túlio Quintino}
\orcid{0000-0003-1332-3477}
\affiliation{Sorbonne Universit\'{e}, CNRS, LIP6, F-75005 Paris, France}
\author{Michał Studziński}
\orcid{0000-0002-5946-9845}
\affiliation{International Centre for Theory of Quantum Technologies, University of Gdańsk, Jana Bażyńskiego 1A, 80-309 Gdańsk, Poland}
\affiliation{Institute of Theoretical Physics and Astrophysics, Faculty of Mathematics, Physics and Informatics, University of Gda\'{n}sk, Wita Stwosza 57, 80-308 Gda\'{n}sk, Poland}
\author{Satoshi Yoshida}
\orcid{0000-0002-0521-5209}
\affiliation{Department of Physics, Graduate School of Science, The University of Tokyo, 7-3-1 Hongo, Bunkyo-ku, Tokyo 113-0033, Japan}
\date{27th April 2026}
\begin{abstract}
This work considers a teleportation task for Alice and Bob in a scenario where Bob cannot perform corrections.  In particular, we analyse the task of \textit{multicopy state teleportation}, where Alice has $k$ identical copies of an arbitrary unknown $d$-dimensional qudit state $\ket{\psi}$ to teleport a single copy of $\ket{\psi}$ to Bob using a maximally entangled two-qudit state shared between Alice and Bob without Bob's correction. Alice may perform a joint measurement on her half of the entangled state and the $k$ copies of $\ket{\psi}$.
We prove that the maximal probability of success for teleporting the exact state $\ket{\psi}$ to Bob is $p(d,k)=\frac{k}{d(k-1+d)}$ and present an explicit protocol to attain this performance.
Then, by utilising $k$ copies of an arbitrary target state $\ket{\psi}$, we show how the multicopy state teleportation protocol can be employed to enhance the success probability of storage and retrieval of quantum programs, which aims to universally retrieve the action of an arbitrary quantum channel that is stored in a state.
Our proofs make use of group representation theory methods, which may find applications beyond the problems addressed in this work.
\end{abstract}

\maketitle


\tableofcontents

\section{Introduction}

Quantum state teleportation is a protocol which allows two parties who share entanglement,
Alice and Bob, to transmit a quantum state without having direct access to a quantum channel~\cite{bennett93}. If Alice and Bob share a maximally entangled state $\ket{\phi_d^+}:=\frac{1}{\sqrt{d}}\sum_{i=0}^{d-1} \ket{ii} \in \mathbb{C}^d \otimes \mathbb{C}^d$ and Alice wants to send an arbitrary unknown qudit state $\ket{\psi}\in\mathbb{C}^d$, she performs a joint measurement on her side, and send the outcome of this measurement to Bob by classical communication. Bob may then perform a correction step, which depends on the outcome of Alice's measurement. In this way, Bob is ensured to hold the state $\ket{\psi}$ by consuming only entanglement and classical communication.
Today, quantum teleportation is a fundamental element ubiquitous in quantum information, being used as a building block of several protocols and tasks~\cite{Briegel1998Repeaters,Gottesman1999Demonstrating,Bennett1996Mixed,nielsen1997programmable,Chitambar2023Duality,murao99,Jozsa05GateTeleportation,Gottesman99GateTeleportation,NielsenChuangBook,wildeShannon}.

Despite its wide variety of applications, the necessary correction step in the standard teleportation protocol has some undesirable consequences. For instance, it forbids us from using the standard teleportation protocol to store a quantum program, i.e.\@, a unitary operation, and later retrieve it to apply on an arbitrary state $\ket{\psi}$~\cite{nielsen1997programmable} since the correction step depends on the quantum program to be stored. To overcome the correction step, Ishizaka and Hiroshima~\cite{ishizaka08,ishizaka09} have proposed a protocol referred to as \textit{Port-Based-Teleportation}  (PBT), where Alice and Bob share $N$ copies of maximally entangled qudit states, and apart from discarding a part of his qudits, Bob does not need to perform any correction. Since its first appearance, PBT has been an active topic of research, and optimal protocols were obtained when Alice wants to teleport arbitrary qudits via an arbitrary number of ports~~\cite{studzinski16}, in a scenario where Alice and Bob may share states which are not maximally entangled~\cite{mozrzymas18}, and also recycling~\cite{strelchuk13}, and a multiport scenario~\cite{studzinski20}. 
Since no corrections on Bob's side are required, PBT allows us to bypass the no-programming theorem~\cite{nielsen1997programmable}, and to perform a probabilistic or deterministic non-exact protocols for programming quantum operations into quantum states, a problem closely related (or even equivalent to some extent) to unitary learning~\cite{bisio10Learning} and storage and retrieval~\cite{sedlak18SAR}. Additionally, the possibility of performing PBT teleportation has found applications in seemly unrelated problems such as unitary estimation~\cite{Yoshida2024one}, unitary inversion~\cite{quintino19PRL,Quintino2021Deterministic}, unitary transposition~\cite{quintino19PRA}, Bell nonlocality~\cite{buhrman16BellPBT}, and nonlocal computation~\cite{Beigi2011Simplified}.

We can view the PBT protocol as a teleportation protocol which, at the cost of consuming entanglement, performs high-quality teleportation without any corrections on Bob's side (apart from discarding parts of his system). A natural question is which other resources or scenarios allow a high-quality correction-free teleportation? Can we improve our performance without consuming extra entanglement? For instance, how can we do a high-quality correction-free teleportation in a scenario where Alice wants to teleport a $d$-dimensional system, and Alice and Bob have a single copy of a $d$-dimensional maximally entangled state, the same amount of entanglement of the standard protocol. In some scenarios, a natural extra resource to consider is the case where Alice has access to \(k\) identical copies of an unknown input state. This scenario arises in a case where Alice has an uncharacterised source that can be used to produce IID (independent and identically distributed) quantum states. Situation encountered in various cryptographic scenarios~\cite{Pirandola2020crypto}.  We note that these copies constitute a nontrivial resource: although they do not reveal the state perfectly for finite \(k\), they can be processed jointly to increase the probability of successful teleportation without requiring any correction on Bob's side. And, even when Alice has access to infinitely many copies (hence, allowing to perform state tomography and having perfect knowledge of her state), as we detail later, the problem remains nontrivial, in the sense that deterministic and exact teleportation without a correction step is not possible.

In this work, we consider a scenario of multicopy state teleportation, where Alice and Bob share a single maximally entangled qudit state, and Alice has access to $k$ identical copies of  of an arbitrary qudit state $\ket{\psi}$ which she desires to teleport a single copy of $\ket{\psi}$ to Bob. We show that, by consuming $k$ copies of $\ket{\psi}$, Alice may teleport $\ket{\psi}$ to Bob with a success probability increasing with the number of copies $k$ in a scenario where Bob does not need to perform any correction. The scenario we consider may be viewed as a probabilistic PBT where Alice and Bob have a single port, but Alice has $k$ copies of the state she desires to teleport. We prove that, when $k$ copies of a $d$-dimensional state are available, the optimal success probability to teleport $\ket{\psi}$ to Bob is $p(d,k)=\frac{k}{d(k-1+d)}$ and present an explicit protocol to attain this performance. As an immediate application, we analyse how multicopy state teleportation may be used to increase the success probability of storage and retrieval of quantum programs (arbitrary quantum channels) when $k$ copies of the desired input state are available.

\section{The multicopy state teleportation task and the main result}
\label{Section2}
\subsection{Motivation and conceptual presentation of the problem }
\begin{figure}[h!]
    \centering
     \includegraphics[width=0.55\textwidth]{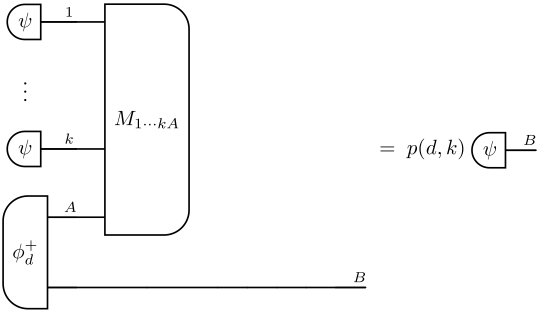}
    \caption{Pictorial illustration of the teleportation scenario considered in this work. Alice and Bob share one pair of a maximally entangled qudit state $\ket{\phi_d^+}$ and Alice has $k$ copies of an arbitrary qudit $\ket{\psi}$. In order to teleport the state $\ket{\psi}$ to Bob, Alice performs a joint measurement $M$ on all quantum states on her side. We consider the case where Bob probabilistically receives a single copy of $\ket{\psi}$ without any corrections, hence Alice's measurement outcome may be assumed to be dichotomic, where one measurement outcome corresponds to success, where Bob receives the exact state $\ket{\psi}$ with probability $p(d,k)$, and a failure case where Bob receives a quantum state which is not $\ket{\psi}$.}
    \label{fig:PBT_many_copies}
\end{figure}

In a standard teleportation protocol, Alice has a pure qudit state $\ket{\psi}\in\mathbb{C}^d$ which she wants to send to Bob in a scenario where the parties do not have access to a quantum channel, but they share one copy of a $d$-dimensional maximally entangled states $\ket{\phi_d^+}:=\frac{1}{\sqrt{d}}\sum_{i=0}^{d-1} \ket{ii} \in \mathbb{C}^d \otimes \mathbb{C}^d$. In this protocol, Alice performs a joint measurement  on her side, and send the outcome of this measurement to Bob by classical communication. Bob may then perform a correction step, which depends on the outcome of Alice's measurement. As discussed in the introduction, in addition to requiring more work, the correction step required by Bob prevents the standard teleportation scheme to be employed for other relevant informational and computational tasks, such as gate programming~\cite{nielsen1997programmable}, unitary storage-and-retrieval~\cite{bisio10Learning,sedlak18SAR}. One way to obtain a high-quality teleportation scheme without any correction is to perform PBT, a protocol which its performance depend on the amount of entanglement shared by the parties. But, when the parties share a single copy of a $d$-dimensional maximally entangled state, PBT reduces to standard teleportation without correction, and it has a success probability of $1/d^2$.

While from a foundations perspective it is common to analyse a situation where one has access to a single copy of a completely unknown quantum state $\ket{\psi}$, in various information and quantum information protocols, one has access to a source which produces multiple copies of a desired state. This is the situation on tasks such as cryptography~\cite{Pirandola2020crypto}, data compression of known states~\cite{Schumacher1995coding} and unknown states~\cite{Yang2016compression}, entanglement distillation~\cite{horodecki_review}, among others. It is then natural to consider the case that Alice does not have classical knowledge of the quantum state $\ket{\psi}$ she holds, but, she has many copies of it. This situation arises for instance, in a scenario where Alice has access to an uncharacterised source which can produce copies of the same quantum state. Alice's goal is then to send this state to Bob in a scenario where the parties do not have a quantum channel, but Alice and Bob share a single copy of a maximally entangled state.  Since our teleportation protocol does not require any correction step, it can be directly used for storing-and-retrieving quantum programs in a scenario where one has access to multiple copies of an input state. This application is detailed in Sec.~\ref{sec:application}.

The multicopy teleportation task is presented in Fig.~\ref{fig:PBT_many_copies}, in a scenario where two parties, Alice and Bob, share a $d$-dimensional qudit maximally entangled state $\ket{\phi^+_d}:=\frac{1}{\sqrt{d}}\sum_i \ket{ii} \in \mathbb{C}^d \otimes \mathbb{C}^d$ and Alice has $k$ copies of an arbitrary qudit state $\ket{\psi}\in\mathbb{C}^d$. Alice's goal is to perform a measurement on her part of the state and prepare the state $\ket{\psi}$ at Bob in a probabilistic heralded manner, that is, with probability $p(d,k)$, Alice knows that the state $\ket{\psi}$ was teleported to Bob perfectly, and with probability $1-p(d,k)$, Alice knows the state $\ket{\psi}$ was not perfectly teleported. 
In a standard quantum state teleportation~\cite{bennett93} scenario, it is always possible for Alice to inform Bob the outcome of her measurement in a way that Bob may perform a correction step. After the appropriate correction, Bob can recover the state $\ket{\psi}$ with $p=1$.
We consider a scenario where Bob cannot do any correction and Alice can only communicate to Bob whether the protocol has worked or failed. 
This goes in a similar direction of \textit{Port-Based-Teleportation} (PBT)~\cite{ishizaka08,ishizaka09,studzinski16,strelchuk13}, where Bob does not perform any correction operation, except for discarding some of his qubits. 

When $k=1$, Alice has a single copy of the arbitrary qudit state $\ket{\psi}$, and the task we consider corresponds to the standard quantum teleportation where Bob cannot perform corrections, which is equivalent to probabilistic PBT with a single port. In this case, the optimal measurement for Alice is any measurement where one of the measurement elements is the projector onto a maximally entangled state $\ketbra{\phi^+_d}$. 
The simplicity of this measurement, together with the lack of correction, makes this correction-less teleportation the most commonly implemented experimentally \cite{boschi98XPteleport}. 
Following the same calculation of standard state teleportation, we see that with probability $p=1/d^2$ the outcome corresponding to $\ketbra{\phi^+_d}$ is obtained, and Bob's state is transformed into $\ketbra{\psi}$. This problem becomes non-trivial when $k>1$, where Alice has then more options for joint measurements to perform on her side. 

Finally, we mention that, by design, independently of the dimension $d$ and number of copies $k$, in the probabilistic multi-copy state teleportation protocol, Alice communicates a single bit to Bob. She just informs Bob if she obtained the successful outcome, hence the protocol worked perfectly, or, if the protocol failed. This is in contrast with the standard teleportation protocol, where Alice must send one out of $d^2$ symbols to Bob (that is, $2$ classical dits). This shows that, one may have non-trivial probabilistic exact teleportation protocols even when Alice has very restricted classical communication to Bob, and motivates the question of quantum teleportation with limited classical communication.

\subsection{Precise mathematical description of the problem}
Consider a scenario where Alice and Bob share a $d$-dimensional qudit maximally entangled state $\ket{\phi^+_d}\in\H_A\otimes\H_B$, where $\H_A\cong \H_B\cong\mathbb{C}^d$ and Alice has $k$ copies of an arbitrary qudit state $\ket{\psi}\in\mathbb{C}^d$, states defined in the linear spaces $\H_i\cong\mathbb{C}^d$, where the index $i$ ranges from $1$ to $k$. The initial state held by Alice and Bob can be described by 
\begin{align}
\ket{\psi}_1\otimes \ket{\psi}_2\otimes\ldots\otimes \ket{\psi}_k\otimes \ket{\phi^+_d}_{AB}\in\H_1\otimes\H_2\otimes\ldots\otimes\H_k\otimes\H_A\otimes\H_B,
\end{align}
where the first $k+1$ subsystems are held by Alice and the last system is held by Bob, see Fig.~\ref{fig:PBT_many_copies}.
Alice then performs a two-outcome quantum measurement, described by the Positive Operator Valued Measure (POVM) \cite{NielsenChuangBook} $\{M,\id-M\}$, where  $M$ is a measurement operator, that is, a linear operator respecting
\begin{align}
 M &\in\L(\H_1\otimes\ldots\otimes\H_k \otimes \H_A) \text{, where $\L(\H)$ is the set of linear operators $\H \rightarrow\H$}\\
 M &\geq 0 \\
 M &\leq \id \text{, because $\id-M\geq0$.}
\end{align}

 In the multicopy quantum state teleportation task, the goal is that, when Alice's measurement succeeds, Bob should hold the state $\ket{\psi}$ on his space $\H_B$ with a success probability of $p(d,k)$, that is,
\begin{align} \label{eq:constraint}
	\tr_{1\ldots kA}\Big(\ketbra{\psi}{\psi}^{\otimes k}_{12\ldots k}\otimes \ketbra{\phi^+_d}{\phi^+_d}_{AB} \; M_{1\ldots kA}\otimes \id_{B}\Big)= p(d,k) \ketbra{\psi}{\psi}_B.
\end{align}
As previously stated, $M\in\mathcal{L}\left((\mathbb{C}^d)^{\otimes (k+1)}\right)$ is a measurement operator, that is, a linear operator respecting  $0 \leq M\leq \id_d^{{\otimes (k+1)}}$, in a way that the set $\{M,\id-M\}$ is a valid POVM, since $M\geq0$, $\id-M\geq0$ and $M+\id-M=\id$. 
We observe that, in the multicopy state teleportation task considered here, the protocol is allowed to fail with probability $1-p(d,k)$, but, when successful, we require the output state to be the \textit{exact same} state as the input state. The exact requirement imposes severe constraints on the measurements performed by Alice. For instance, most of the possible measurements, including the case where $M$ is the projector onto the symmetric space, would not be able to satisfy Eq.~\eqref{eq:constraint} for every state $\ket{\psi}$. 

In this work, we will present the measurement operator $M_{1\ldots kA}\in \L(\H_1\otimes\ldots\otimes\H_k \otimes \H_A)$ which maximises the success probability $p(d,k)$ of Bob obtaining $\ket{\psi}$ perfectly. In more precise terms, we solve the following optimisation problem,
\begin{align}
    \text{Given } d,k \in \mathbb{N}, \nonumber \\ 
    &\max_{\mathllap{M\in\mathcal{L}}\mathrlap{\cramped{\left({\mathbb{C}^d}\right)^{\otimes(k+1)\phantom{^2}}}}} \; p(d,k)\in[0,1] \label{eq:mainSDP}\\ 
    \text{such that: } & \forall \ket{\psi}\in\mathbb{C}^d \text{ with } \norm{\ket{\psi}}=1  \\
    &\tr_{1\ldots kA}\Big(\ketbra{\psi}{\psi}^{\otimes k}_{12\ldots k}\otimes \ketbra{\phi^+_d}{\phi^+_d}_{AB} \; M_{1\ldots kA}\otimes \id_{B}\Big)= p(d,k) \ketbra{\psi}{\psi}_B, \label{eq:notSDP} \\
    &  0 \leq M\leq \id .
\end{align}
In Sec.~\ref{sec:Proofs}, we show how to rewrite the above problem as a Semidefinite Program (SDP). In this current version, 
it is not immediately a SDP, since Eq.~\eqref{eq:notSDP} has to hold for all normalised vectors $\ket{\psi}$, the problem has infinitely many constraints, and it is not an SDP. 
However, due to linearity, we can reduce these infinitely many constraints to a finite number of them, and by making use of the symmetries in this problem, 
we solve the optimisation problem and obtain our main result:  the optimal measurement $M$ and its probability $p(d,k)$ for every $d,k \in \mathbb{N}$.
\begin{restatable}[Main Result]{theorem}{MainThm}
\label{thm:1}
The maximal success probability in the multicopy state teleportation problem described in Eq.~\eqref{eq:mainSDP} with $k$ copies of the arbitrary qudit state $\ket{\psi}\in\mathbb{C}^d$ is given by
\begin{align}
\label{eq:achev}
	p(d,k)=\frac{k}{d(k-1+d)}.
\end{align}
The maximal success probability is attainable by the POVM element
\begin{align}
\label{eq:achev2}
    M_{1\ldots kA}= \frac{dk}{(k-1+d)}\Big(P^{sym}_{1\ldots k}\otimes \id_A\Big)\Big(\id_{1\ldots(k-1)}\otimes \ketbra{\phi_d^+}_{kA} \Big)\Big(P^{sym}_{1\ldots k}\otimes \id_A\Big),
\end{align}
where $P^{sym}_{1\ldots k}$ is the projector onto the symmetric subspace%
\footnote{A vector $\ket{\psi}\in\left(\mathbb{C}^d\right)^{\otimes k}$, belongs to the symmetric subspace of $\left(\mathbb{C}^d\right)^{\otimes k}$ if it is party permutation invariant, that is, for any permutation $\pi\in S_k$ of a set with $k$ elements, we have $\ket{\psi}_{12\ldots k}=\ket{\psi}_{\pi(12\ldots k)}$~\cite{WatrousBook,Harrow2013Church}.}%
of $\left(\mathbb{C}^d\right)^{\otimes k}$, and acts on systems $1,\ldots,k$, while the identity operator $\id_{1\ldots (k-1)}$ acts on systems $1,\ldots, k-1$.
\end{restatable}
The proof of Thm.~\ref{thm:1} is presented in Sec.~\ref{sec:Proofs}. For now, we notice that, as it should be, when a single copy of the arbitrary {state} $\ket{\psi}\in\mathbb{C}^d$ is available, we obtain $p_s(d,k=1)=\frac{1}{d^2}$, which is the optimal success probability in the standard teleportation scenario without correction. Also, when $k$ is very large, we obtain $p(d,k\to\infty)=\frac{1}{d}$, which is the performance of Remote State Preparation (RSP)~\cite{Bennett01RSP}. One interpretation of the RSP is that Alice performs a teleportation protocol with a known state $\ket{\psi}$, in contrast with the standard quantum teleportation protocol where $\ket{\psi}$ is an unknown arbitrary state. From this perspective, we may understand the formula  $p(d,k\to\infty)=\frac{1}{d}$ as, when Alice has infinitely many copies, she may first perform quantum state tomography, and then implement the RSP protocol to obtain a success probability of $\frac{1}{d}$.

In the Appendix~\ref{appendix}, we discuss the relationship between multicopy state teleportation and the concept of probabilistic simulation of quantum channel from the future to the past proposed at Ref.~\cite{genkina2012optimal}. More precisely, we stablish a lemma that allows a strong connection between our work and Ref.~\cite{genkina2012optimal}, this lemma then leads to an an alternative proof of our Thm.~\ref{thm:1}.

\section{Application: storage and retrieval of quantum programs when $k$ copies of the input state are available} \label{sec:application}
\begin{figure}[h!]
    \centering
    \includegraphics[width=.75\linewidth]{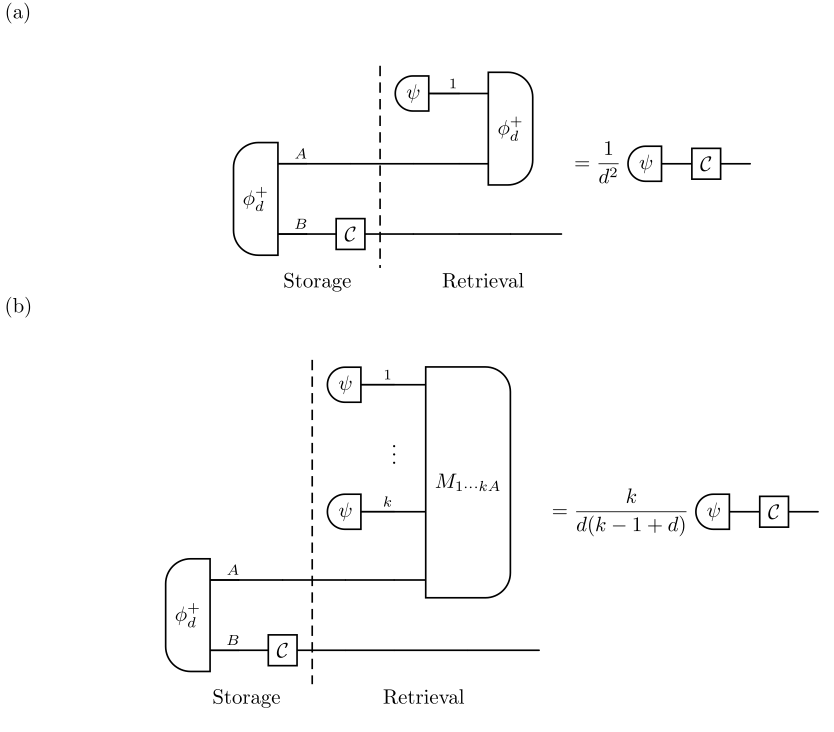}
\caption{  (a) The task of storage and retrieval of quantum programs, which can be used to universally retrieve the action of an arbitrary quantum channel $\mathcal{C} : \mathcal{L}(\mathbb{C}^d) \to \mathcal{L}(\mathbb{C}^{d'})$ on an arbitrary state $\ket{\psi} \in \mathbb{C}^d$ with probability $p(d, k - 1) = \frac{1}{d^2}$. \\
(b) Employing the multicopy state teleportation protocol to retrieve the action of an arbitrary operation $\mathcal{C} : \mathcal{L}(\mathbb{C}^d) \to \mathcal{L}(\mathbb{C}^{d'})$ when $k$ copies of an arbitrary state $\ket{\psi} \in \mathbb{C}^d$ are available, with probability $p(d, k) = \frac{k}{d(k - 1 + d)}$.}
    \label{fig:multiCopySAR}
\end{figure}
The standard state teleportation protocol may be used to teleport quantum \textit{operations}, which in this context are also referred to as a \textit{quantum programs}, in a scheme sometimes phrased as gate-teleportation~\cite{Jozsa05GateTeleportation,Gottesman99GateTeleportation}, a technique which may be used to perform quantum unitary programming~\cite{nielsen1997programmable} and probabilistic unitary storage and retrieval~\cite{sedlak18SAR}. 
The idea works as follows (see Fig.~\ref{fig:multiCopySAR}), let $U\in\L(\mathbb{C}^d)$ be an arbitrary (unknown) unitary operation, which may be implemented by a whole quantum circuit, referred to as a unitary quantum program $U$. We may then make a single use of this operation and apply it into a half of a maximally entangled qudit state $\ket{\phi_d^+}_{AB}\in\mathbb{C}^d\otimes\mathbb{C}^d$, to obtain the state $\id_A \otimes U_B \ket{\phi_d^+}_{AB}$, which will be stored until the user desires to retrieve the use of the operation $U$. 
When the user desires to obtain the operation $U$ on an arbitrary state $\ket{\psi}\in\mathbb{C}^d$, the user performs a measurement with a POVM element $M=\ketbra{\phi_d^+}_{1A}$ on the state $\ket{\psi}_1\otimes\ket{\phi_d^+}_{AB}$. It follows from the same calculation used in the teleportation trick of Eq.~\eqref{eq:TeleportationTrick}, that the state $U\ket{\psi}$ is obtained in the space $\H_B$ with probability $\frac{1}{d^2}$. 
Note that we may drop the assumption that the operation considered is unitary. An analogous argument shows that we can teleport the action of non-unitary quantum programs as well. The concept of a general deterministic program is formalised by a quantum channel~\cite{WatrousBook,wildeShannon}, which is a completely positive trace-preserving (CPTP) linear map $\mathcal{C} : \mathcal{L}(\mathbb{C}^d) \to \mathcal{L}(\mathbb{C}^{d'})$ that transforms qudit states in $\mathbb{C}^d$ into quantum states in $\mathbb{C}^{d'}$. The protocol for storage and retrieval of arbitrary quantum programs follows the same steps as in the unitary case. In the storage phase, we apply the arbitrary (unknown) channel $\mathcal{C}$ on one half of a maximally entangled qudit state $\ket{\phi_d^+} \in \mathbb{C}^d \otimes \mathbb{C}^d$, obtaining the state $\id\otimes \mathcal{C} \left( \ketbra{\phi_d^+} \right)$. Later, in the retrieval phase,  the same measurement with $M=\ketbra{\phi_d^+}_{1A}$ on the resulting state allows us to universally retrieve the action of $\mathcal{C}$ on an arbitrary state $\ket{\psi}$ with probability $\frac{1}{d^2}$ see Fig.~\ref{fig:multiCopySAR}).

We may now use multicopy state teleportation in the retrieval step described in the previous paragraph. The idea now is, if instead of having a single copy of the arbitrary qudit state $\ket{\psi}\in\mathbb{C}^d$, we have $k$ copies of it, that is, we hold $\ket{\psi}^{\otimes k}$. Instead of performing the retrieval operation $M_{1A}=\ketbra{\phi_d^+}$, we perform the operation $M_{1\ldots kA}:=\frac{dk}{(k-1+d)}\Big(P^{sym}_{1\ldots k}\otimes \id_A\Big)\left( \id_{1\ldots k}\otimes \ketbra{\phi_d^+}_{kA} \right)\Big(P^{sym}_{1\ldots k}\otimes \id_A\Big)$ on $\ket{\psi}^{\otimes k}_{1\ldots k}$ and Alice's part of the maximally entangled state $\ket{\phi_d^+}_{AB}$ to obtain the output state $U\ket{\psi}_B$ with probability $p(d,k)=\frac{k}{d(k-1+d)}$. 
Similarly, as in the case where a single copy of $\ket{\psi}$ is available, we may drop the assumption that the operation is unitary. More precisely, we can perform a single application of an arbitrary (unknown) channel $\mathcal{C} : \mathcal{L}(\mathbb{C}^d) \to \mathcal{L}(\mathbb{C}^{d'})$ on one half of a maximally entangled state to obtain the state $\id \otimes \mathcal{C} \left( \ketbra{\phi_d^+} \right)$. Then, if the measurement $M_{1 \ldots k A}$ is performed on the first $k + 1$ systems, as in Fig.~\ref{fig:multiCopySAR}, the state $\mathcal{C}(\ketbra{\psi})$ is teleported to Bob with probability $p(d, k) = \frac{k}{d(k - 1 + d)}$.

As discussed earlier, when $k\to\infty$, we obtain $p(d,k\to\infty)=\frac{1}{d}$, which corresponds to the case where we want to retrieve the operation $U$ stored in  $\id\otimes U \ket{\phi^+_d}$ in a ``known'' input state $\ket{\psi}$ instead of an arbitrary unknown one. As also discussed earlier, when one has access to infinitely many copies of the state, one can perform quantum state tomography and know exactly the state $\ket{\psi}$, hence the measurement $M$ may depend explicitly on $\ket{\psi}$. When $k\to\infty$, the task presented in this section is in a close relationship to remote state preparation~\cite{Bennett01RSP}, and may be viewed as a variation of the storage and retrieval task~\cite{sedlak18SAR} in a scenario where the state on which one desires to retrieve the operation is known. Also, when restricted to the case where $k\to\infty$, the task presented here is equivalent to the single-port version of the Port-Based State Preparations task presented in Ref.~\cite{Muguruza2024PBRSP}, and the problem of retrieving the usage of multiple calls of a quantum operation on known quantum states is also analysed in Ref.~\cite{Brizc2024PBRSP} in the context of higher-order quantum computing with known input states.

\section{Proof of attainability and the intuition behind the protocol}
We restate here Eq.~(\ref{eq:achev2}) from Thm.~\ref{thm:1}, 
which gives the POVM element describing the optimal measurement:
\begin{align}
    M_{1\ldots kA}=\frac{dk}{(k-1+d)}\Big(P^{sym}_{1\ldots k}\otimes \id_A\Big)\Big(\id_{1\ldots(k-1)}\otimes \ketbra{\phi_d^+}_{kA} \Big)\Big(P^{sym}_{1\ldots k}\otimes \id_A\Big).
    \tag{\ref{eq:achev2}}
\end{align}
Before presenting in Sec.~\ref{sec:Proofs} a detailed proof that this is indeed the optimal measurement,  we analyse some properties of this measurement to give an intuition on why this measurement is useful for multicopy state teleportation, and prove that the success probability $p(d,k)=\frac{k}{d(k-1+d)}$ is attainable.

First, notice that $M_{1\ldots kA}=\frac{dk}{(k-1+d)}\Big(P^{sym}_{1\ldots k}\otimes \id_A\Big)\Big(\id_{1\ldots(k-1)}\otimes \ketbra{\phi_d^+}_{kA} \Big)\Big(P^{sym}_{1\ldots k}\otimes \id_A\Big)$ generalises the measurement performed by Alice in the standard probabilistic teleportation protocol, i.e., when $k=1$, we obtain $M_1 = \ketbra{\phi_d^+}_{1A}$. Also, it is not hard to show that, if we set Alice's measurements to be $ M_{1\ldots kA}=\frac{dk}{(k-1+d)}\Big(P^{sym}_{1\ldots k}\otimes \id_A\Big)\Big(\id_{1\ldots(k-1)}\otimes \ketbra{\phi_d^+}_{kA} \Big)\Big(P^{sym}_{1\ldots k}\otimes \id_A\Big)$, the multicopy state teleportation protocol is successful with probability $p(d,k)=\frac{k}{d(k-1+d)}$. 
By successively using 
the cyclic property of the trace $\tr\big(ABC\big)=\tr\big(BCA\big)$,
the identity $P^{sym}_{1\ldots k}\ket{\psi}^{\otimes k}_{12\ldots k} = \ket{\psi}^{\otimes k}_{12\ldots k}$,
and the equality derived by the ``teleportation trick''~\cite{bennett93} 
\begin{align}    
\tr_{1A}\biggl(\Big(\ketbra{\phi^+_d}_{1A} \otimes \id_B \Big)\Big( \ketbra{\psi}_1 \otimes \ketbra{\phi^+_d}_{AB}\Big)\biggl)
=
\frac{1}{d^2} \ketbra{\psi}_B,
\end{align}
one may check that,
\small
\begin{align}
  \tr_{1\ldots kA}
    &\biggl(\ketbra{\psi}{\psi}^{\otimes k}_{12\ldots k}\otimes \ketbra{\phi^+_d}{\phi^+_d}_{AB} \; M_{1\ldots kA}\otimes \id_{B}\biggl) \\
    &\begin{multlined}[c][.87\linewidth]=\frac{dk}{(k-1+d)}  \tr_{1\ldots kA}\biggl(\ketbra{\psi}{\psi}^{\otimes k}_{12\ldots k}\otimes \ketbra{\phi^+_d}{\phi^+_d}_{AB} \;\Big(P^{sym}_{1\ldots k}\otimes \id_{A}\Big)
       \\\Big(\id_{1\ldots(k-1)}\otimes \ketbra{\phi_d^+}_{kA} \Big)
       \Big(P^{sym}_{1\ldots k}\otimes \id_{A}\Big) \otimes \id_{B}\biggl)\end{multlined}\\
    &\begin{multlined}[c][.87\linewidth]=\frac{dk}{(k-1+d)}  \tr_{1\ldots kA}\biggl(\Big(P^{sym}_{1\ldots k}\otimes \id_{AB}\Big) \ketbra{\psi}{\psi}^{\otimes k}_{12\ldots k}\otimes \ketbra{\phi^+_d}{\phi^+_d}_{AB} \;\Big(P^{sym}_{1\ldots k}\otimes \id_{AB}\Big)\\\Big(\id_{1\ldots(k-1)}\otimes \ketbra{\phi_d^+}_{kA} \otimes \id_B \Big)\biggl)\end{multlined}\\
    &=\frac{dk}{(k-1+d)}  \tr_{1\ldots kA}\biggl(\ketbra{\psi}{\psi}^{\otimes k}_{12\ldots k}\otimes \ketbra{\phi^+_d}{\phi^+_d}_{AB} \Big(\id_{1\ldots(k-1)}\otimes \ketbra{\phi_d^+}_{kA} \otimes \id_B \Big)\biggl)\\
    &=  \frac{dk}{(k-1+d)} \frac{1}{d^2} \ketbra{\psi}{\psi}_B \\
    &=  \frac{k}{d(k-1+d)} \ketbra{\psi}{\psi}_B.
\end{align}
\normalsize
Finally, since $M_{1\ldots kA}$ is the composition of positive semidefinite operators, we have that $M_{1\ldots kA}\geq0$. In order to finish the attainability proof, one must also show that $ M_{1\ldots kA}\leq \id$, of which detailed proof is shown in Section~\ref{sec:Proofs}. 

In the next subsections, we present an alternative approach to the problem which, in addition to providing an intuition on why the proposed measurements are optimal, it will lead into an eigendecomposition of the measurement operator  $M_{1\ldots kA}$ from Eq.~(\ref{eq:achev2}). 

\subsection{The $k=2$ case}
Before start analysing the $k=2$ case, let us quickly revisit standard state teleportation, i.e., the $k=1$ case.
When $k=1$, Alice's measurement operator is given by $M_{1A}=\ketbra{\phi_d^+}_{1A}$, Alice and Bob share the state $\ket{\psi}_1\otimes\ket{\phi_d^+}$, and the teleportation protocol works thanks to the ``teleportation trick'' calculation
\begin{align}\label{eq:TeleportationTrick}
    \Big(\bra{\phi_d^+}_{1A}\otimes\id_B\Big) \Big(\ket{\psi}_1\otimes\ket{\phi_d^+}_{AB}\Big) = \frac{1}{d}\ket{\psi}_B.
\end{align}
This calculation is shown by recalling that $\ket{\phi_d^+}:=\frac{1}{\sqrt{d}}\sum_i \ket{ii}$ and that an arbitrary state can always be decomposed in the computational basis as $\sum_i \braket{i}{\psi}\ket{i}$, and then we have
\begin{align}\label{eq:TeleportationTrick2}
    \Big(\bra{\phi_d^+}_{1A}\otimes\id_B\Big) \Big(\ket{\psi}_1\otimes\ket{\phi_d^+}_{AB}\Big) =&
    \frac{1}{d}\sum_{ikl} \Big(\bra{kk}_{1A}\otimes \id_B\Big)\Big(\braket{i}{\psi} \ket{i}_1\otimes \ket{ll}_{AB}\Big) \\
    =&
    \frac{1}{d}\sum_{ikl} \braket{k}{i}_{1}\braket{k}{l}_{A} \braket{i}{\psi} \ket{l}_B \\
    =&
    \frac{1}{d}\sum_{i} \braket{i}{\psi} \ket{i}_B \\
        =&
    \frac{1}{d} \ket{\psi}_B.
\end{align}
Rewritten in density matrix notation, the teleportation trick of Eq.~\eqref{eq:TeleportationTrick}
\begin{align}
    \tr_{1A}\biggl(\Big(\ketbra{\phi_d^+}_{1A}\otimes\id_B \Big)\Big(\ketbra{\psi}_1\otimes\ketbra{\phi_d^+}_{AB}\Big)\biggl) = \frac{1}{d^2}\ketbra{\psi}_B,
\end{align}
is obtained. It may also be shown by making use of the identity $A\otimes \id \ket{\phi_d^+}=\id \otimes A^T \ket{\phi_{d'}^+}$ which holds for any linear operator $A:\mathbb{C}^d\to \mathbb{C}^{d'}$, and $A^T$ represents the transposition of $A$ in the computational basis.

We may now try to imagine a way to generalise Alice's measurement for $k=2$, for that, we seek for states which is permutation invariant in the spaces $\H_{1}\otimes\H_{2}$, and has the teleportation property of $\phi_d^+$ in the space $A$. One idea to attain this goal is to define the states:
\begin{align}
    \ket{r_i}:=\frac{\ket{i}_1\otimes \ket{\phi_d^+}_{2A} +  \ket{i}_2\otimes \ket{\phi_d^+}_{1A}}{\gamma}, \quad i\in\{1,\ldots, d\}
\end{align}
where $\gamma\in\mathbb{R}$ is a normalisation factor to ensure $\norm{\ket{r_i}}=1$. Notice that, for any state $\ket{\psi}\in\mathbb{C}^d$, by using the teleportation trick from Eq.~\eqref{eq:TeleportationTrick}, we have\footnote{When it is apparent} from the context, we may suppress the tensor product symbol $\otimes$, that is we may use the notation given by $\ket{\psi}_1\ket{\psi}_2:=\ket{\psi}_1\otimes\ket{\phi}_2$.
\begin{align}
    \Big(\bra{r_i}_{12A}\otimes \id_B\Big) \Big( \ket{\psi}_1\ket{\psi}_2\ket{\phi_d^+}_{AB}\Big)&=
\frac{1}{\gamma}\Big(\bra{i}_1 \bra{\phi_d^+}_{2A}\otimes \id_B+   \bra{i}_2 \bra{\phi_d^+}_{1A}\otimes \id_B\Big) \Big( \ket{\psi}_1\ket{\psi}_2\ket{\phi_d^+}_{AB}\Big)\\
&=\frac{1}{\gamma}\left(\frac{1}{d} \braket{i}{\psi}_1 \ket{\psi}_B + \frac{1}{d} \braket{i}{\psi}_2 \ket{\psi}_B\right) \\
&=\frac{2}{\gamma d} \braket{i}{\psi} \ket{\psi}_B.
\end{align}
Then if we set $M_{12A}=\sum_{i=1}^d\ketbra{r_i}_{12A}$, we have
\begin{align}
    \sum_{i=1}^d \tr_{12A}\Big( \ketbra{r_i}_{12A}\otimes \id_B \; \ketbra{\psi}_1\otimes\ketbra{\psi}_2\otimes \ketbra{\phi_d^+}_{AB} \Big)
    &= \sum_i \frac{2}{\gamma d} \frac{2}{\gamma d} \braket{i}{\psi}\braket{\psi}{i} \ketbra{\psi}_B \\ 
    &= \frac{4}{\gamma^2 d^2} \left( \sum_i  \abs{\braket{i}{\psi}}^2 \right)  \ketbra{\psi}_B \\
    &= \frac{4}{\gamma^2 d^2} \ketbra{\psi}_B,
\end{align}
where the identity $\sum_i  \abs{\braket{i}{\psi}}^2=\norm{\ket{\psi}}=1$ follows from the fact that $\ket{\psi}$ is a quantum state, hence a normalised vector.

In direct terms, if we prove that the vectors $\ket{r_i}$ are orthonormal and the constant $\gamma$ is calculated, it is enough to set the measurement operator as $M_{12A}=\sum_{i=1}^d\ketbra{r_i}_{12A}$ to obtain perfect teleportation with probability $p(d,k=2)=\left(2/\gamma d\right)^2$. As we will verify soon,
Lemma~\ref{lemma:orthonormality} ensures that the vectors $\ket{r_i}$ are orthonormal and that $\frac{1}{\gamma}=\sqrt{\frac{d}{k(k-1+d)}}$ holds.

\subsection{The $k=3$ case} 
For the $k=3$, we aim to generalise the states $\ket{r_i}:=\frac{\ket{i}_1\otimes \ket{\phi_d^+}_{2A} +  \ket{i}_2\otimes \ket{\phi_d^+}_{1A}}{\gamma}, \quad i\in\{1,\ldots, d\}$, in a way that it is permutation invariant in the spaces $\H_1\otimes\H_2\otimes\H_3$.
One idea is to set
\begin{align}
    \ket{r_i}:=\frac{\ket{s_i}_{12}\otimes \ket{\phi_d^+}_{3A} +  \ket{s_i}_{13}\otimes \ket{\phi_d^+}_{2A} + \ket{s_i}_{23}\otimes \ket{\phi_d^+}_{1A}}{\gamma}, \quad i\in\left\{1,\ldots, {d\choose 2}\right\}
\end{align}
where $\{\ket{s_i}\}_{i=1}^{d \choose 2}$ be an orthonormal basis for the symmetric subspace%
\footnote{The symmetric subspace of $\left(\mathbb{C}^d\right)^{\otimes {k}}$ has dimension ${k-1+d\choose k}$, and an orthonormal basis for this space may be obtained by a ``type'' approach and other standard techniques~\cite{WatrousBook,Harrow2013Church}.}
of $\left(\mathbb{C}^d\right)^{\otimes2}$. Following the same steps of the $k=2$ case, we see that
\begin{align}
    \Big(\bra{r_i}_{123A}\otimes \id_B\Big) \Big( \ket{\psi}_1\ket{\psi}_2\ket{\psi}_3\ket{\phi_d^+}_{AB}\Big) 
    &= \frac{3}{\gamma d} \ket{\psi}_B \bra{s_i}\Big(\ket{\psi}\otimes\ket{\psi}\Big),
\end{align}
and then
\begin{align}
    \sum_{i=1}^{d \choose 2} \tr_{123A}\Big( \ketbra{r_i}_{123A}\otimes \id_B \; &\ketbra{\psi}_1\otimes\ketbra{\psi}_2\otimes\ketbra{\psi}_3\otimes \ketbra{\phi_d^+}_{AB} \Big)\\
    =& \sum_i \frac{3}{\gamma d} \frac{3}{\gamma d} \Big(\bra{s_i}_{12}\ket{\psi}_1\ket{\psi}_2\Big) \Big(\bra{\psi}_1\bra{\psi}_2\ket{s_i}_{12}\Big) \ketbra{\psi}_B \\ 
    =& \frac{9}{\gamma^2 d^2} \left( \sum_i  \abs{\bra{s_i}_{12}\,\ket{\psi}_1\ket{\psi}_2}^2 \right)  \ketbra{\psi}_B \\
    =& \frac{9}{\gamma^2 d^2} \ketbra{\psi}_B,
\end{align}
where the identity $\sum_{i=1}^{d \choose 2} \abs{\bra{s_i}_{12}\,\ket{\psi}_1\ket{\psi}_2}^2=1$ follows from the fact that the state $\ket{\psi}\otimes\ket{\psi}$ belongs to the symmetric subspace of%
\footnote{To prove that, just notice that $\ket{\psi}\otimes\ket{\psi}$ belongs to the symmetric subspace of $\left(\mathbb{C}^d\right)^{\otimes2}$, we can write $\ket{\psi}\otimes\ket{\psi}$  as a linear combination of the vectors $\{\ket{s_i}\}_i$.}%
$\left(\mathbb{C}^d\right)^{\otimes2}$.

\subsection{The $k\in\mathbb{N}$ case} \label{sec:Arbitrary_k}
We are now in a position to define the general form of the states $\ket{r_i}$ which will form an orthonormal basis leading to Alice's measurement.
\begin{definition} \label{def:r_i}
Let $\{\ket{s_i}\}_{i=1}^{k-2+d \choose k-1}$ be an orthonormal basis for the symmetric subspace of $\left(\mathbb{C}^d\right)^{\otimes(k-1)}$, $V_{(ak)}:\left(\mathbb{C}^d\right)^{\otimes k} \to \left(\mathbb{C}^d\right)^{\otimes k} $ be the operator which swaps the system in position $a$ with the system in position $k$.
For any $k,d\in \mathbb{N}$ and $i\in\{1,\ldots, {k-2+d \choose k-1}\}$, we define the vectors $\ket{r_i}\in\left(\mathbb{C}^d\right)^{\otimes(k+1)}$ as
\begin{align}
    \ket{r_i}:=\sqrt{\frac{d}{k(k-1+d)}} \left( \sum_{a=1}^{k} V_{(ak)}\otimes \id_A \Big(\ket{s_i}_{12\ldots(k-1)} \otimes \ket{\phi_d^+}_{kA} \Big)\right),
\end{align}
where $\ket{\phi_d^+}:=\frac{1}{\sqrt{d}}\sum_{i=0}^{d-1}\ket{ii}$ is the maximally entangled qudit state.
\end{definition}

We start by proving that the vectors $\ket{r_i}$ presented in Def.~\ref{sec:Arbitrary_k} are indeed orthonormal.

\begin{lemma} \label{lemma:orthonormality}
    The vectors $\ket{r_i}\in\left(\mathbb{C}^d\right)^{\otimes(k+1)}$ presented in Def.~\ref{def:r_i} are orthonormal, that is $\braket{\phi_i}{\phi_j}=\delta_{ij}$.
\end{lemma}
\begin{proof}
    First, notice that {it} follows from the ``teleportation trick'' presented in Eq.~\eqref{eq:TeleportationTrick} that for any vector $\ket{s_j}_{2\ldots k}$, it holds that
\begin{align}
    \Big( \id_{12\ldots(k-1)} \otimes \bra{\phi_d^+}_{kA}\Big)  \Big( \ket{s_j}_{2\ldots k}\otimes \ket{\phi_d^+}_{1A}\Big)
    = \frac{1}{d} \ket{s_j}_{2\ldots A}.
\end{align}
    Now, let us focus on the vectors $\ket{\Xi(a,i)}:=V_{(ak)}\otimes \id_A \Big(\ket{s_i}_{12\ldots(k-1)} \otimes \ket{\phi_d^+}_{kA} \Big)$. Direct calculation shows that $\braket{\Xi(a,i)}{{\Xi(a,j)}}=\delta_{ij}$ holds. But, when $a\neq a'$, we have $\braket{\Xi(a,i)}{{\Xi(a',j)}}=\frac{\delta_{ij}}{d}$. To show that, let us first fix $a=k$ and $a'=1$ to see that
\begin{align}
   \braket{\Xi(a,i)}{{\Xi(a',j)}}=&\Big( \bra{s_i}_{12\ldots(k-1)} \otimes \bra{\phi_d^+}_{kA}\Big)  \Big( \ket{s_j}_{2\ldots k} \otimes \ket{\phi_d^+}_{1A}\Big) \\
   =& \Big(\bra{s_i}_{12\ldots(k-1)} \Big)\Big( \id_{12\ldots(k-1)} \otimes \bra{\phi_d^+}_{kA}\Big)  \Big( \ket{s_j}_{2\ldots k}\otimes \ket{\phi_d^+}_{1A}\Big) \\
   =& \frac{1}{d}\Big( \bra{s_i}_{A2\ldots(k-1)}\Big) \Big( \ket{s_j}_{2\ldots A} \Big) \\
   =&\frac{\delta_{ij}}{d}.
\end{align}
Hence, we have that 
\begin{align}
    \braket{r_i}{r_j} & = \frac{d}{k(k-1+d)} \left(\sum_{a=1}^k \sum_{a'=1}^k  \braket{\Xi(a,i)}{{\Xi(a',j)}} \right)\\ 
    & = \frac{d}{k(k-1+d)} \; \left(\left(\sum_{a=1}^k \delta_{ij} \right) + \left( \sum_{a\neq a' } \frac{\delta_{ij}}{d}\right)\right) \\ 
    & = \delta_{ij}\frac{d}{k(k-1+d)} \; \left(k +  \frac{k(k-1)}{d}\right) \\ 
    & = \delta_{ij}\frac{d}{k(k-1+d)} \; \left( \frac{kd+ k(k-1)}{d}\right) \\ 
    & = \delta_{ij}\frac{1}{(k-1+d)} \; \left( \frac{d+ k-1}{1}\right) \\ 
    & = \delta_{ij},
\end{align}
which concludes the proof.
\end{proof}

Now, we prove that if we perform a quantum measurement using the orthonormal basis $\{\ket{r_i}\}_i$, we perform multicopy state teleportation with a success probability of $p(d,k)=\frac{k}{d(k-1+d)}$, as presented in Thm.~\ref{thm:1}.
\begin{restatable}{lemma}{Attainability}
\label{lemma:attainability}
If we set 
\begin{align}
    M_{1\ldots kA} = \sum_{i=1}^{k-2+d \choose k-1} \ketbra{r_i}_{1\ldots kA},
\end{align}
it holds true that, for any normalised vector $\ket{\psi}\in\mathbb{C}^d$ and any $d,k\in\mathbb{N}$ we have
\begin{align} \label{eq:psLemma}
    \tr_{1\ldots kA}\Big(\ketbra{\psi}{\psi}^{\otimes k}_{12\ldots k}\otimes \ketbra{\phi^+_d}{\phi^+_d}_{AB} \; M_{1\ldots kA}\otimes \id_{B}\Big)= p(d,k) \ketbra{\psi}{\psi}_B
\end{align}
    with 
    \begin{align}
	p(d,k)=\frac{k}{d(k-1+d)}.
\end{align}
\end{restatable}
\begin{proof}
The proof follows from similar steps of the calculations made in previous sections with $k=2$ and $k=3$. For convenience, let us define the vector
$\ket{\Xi(a,i)}:=V_{(ak)}\otimes \id_A \Big(\ket{s_i}_{12\ldots(k-1)} \otimes \ket{\phi_d^+}_{kA} \Big)$ and evaluate the quantity
\begin{align}
    &\Big(\bra{\Xi(a=k,i)}_{1\ldots k A} \otimes \id_B \Big) \ket{\psi}_1\ldots\ket{\psi}_k\ket{\phi_d^+}_{AB}  \\
    =& \Big(\bra{s_i}_{12\ldots(k-1)} \otimes \bra{\phi_d^+}_{kA} \Big) \ket{\psi}_1\ldots\ket{\psi}_k\ket{\phi_d^+}_{AB}\\
    =&    \Big(\bra{s_i}_{12\ldots(k-1)}\;\ket{\psi}_1\ldots\ket{\psi}_{k-1}\Big) \Big(\bra{\phi_d^+}_{kA} \otimes \id_B \Big) \ket{\psi}_k\ket{\phi_d^+}_{AB}\\
    =& \frac{1}{d}  \Big(\bra{s_i}_{12\ldots(k-1)}\;\ket{\psi}_1\ldots\ket{\psi}_{k-1}\Big) \ket{\psi}_B \\
\end{align}
Analogous calculation shows that, for every $a\in\{1,\ldots, k\}$ and every $i\in\left\{1,\ldots { k-2+d \choose k-1} \right\}$, it holds that 
\begin{align}
    \Big(\bra{\Xi(a,i)}_{1\ldots k A} \otimes \id_B \Big) \ket{\psi}_1\ldots\ket{\psi}_k\ket{\phi_d^+}_{AB} 
    =
     \frac{1}{d}  \Big(\bra{s_i}_{12\ldots(k-1)}\;\ket{\psi}_1\ldots\ket{\psi}_{k-1}\Big) \ket{\psi}_B,
\end{align}
and since $\ket{\psi}_1\ldots\ket{\psi}_{k-1}$ belongs to the symmetric subspace of $({\mathbb{C}^d})^{\otimes (k-{1})}$, it holds that 
\begin{align}
\sum_{i=1}^{k-2+d \choose k-1} \bra{s_i}_{12\ldots(k-1)}\;\ket{\psi}_1\ldots\ket{\psi}_{k-1} = 1,
\end{align} hence
\begin{align}
    \sum_{a=1}^k \frac{1}{d} \ket{\psi}_B \sum_{i=1}^{k-2+d \choose k-1} \bra{s_i}_{12\ldots(k-1)}\;\ket{\psi}_1\ldots\ket{\psi}_{k-1}
    =\frac{k}{d}\ket{\psi}_B.
\end{align}
    We can then finish the proof by direct calculation,
\begin{align}
    &\tr_{1\ldots kA}\Big(\ketbra{\psi}{\psi}^{\otimes k}_{12\ldots k}\otimes \ketbra{\phi^+_d}{\phi^+_d}_{AB} \; M_{1\ldots kA}\otimes \id_{B}\Big)\\
    =&\sum_{i=1}^{k-2+d \choose k-1} \tr_{1\ldots kA}\Big(\ketbra{\psi}{\psi}^{\otimes k}_{12\ldots k}\otimes \ketbra{\phi^+_d}{\phi^+_d}_{AB} \; \ketbra{r_i}_{1\ldots k,A}\otimes \id_{B}\Big) \\
    =&\frac{d}{k(k-1+d)}  \frac{k}{d^2} \ketbra{\psi}_B \\ 
    =&\frac{1}{d(k-1+d)}  \ketbra{\psi}_B.
\end{align}
\end{proof}
In order to have another perspective on why the measurements $M_{1\ldots kAB}$ attain the success probability of $p(d,k)=\frac{1}{d(k-1+d)}$, in Sec.~\ref{sec:alternative} we present an alternative proof of Lemma~\ref{lemma:attainability} which makes use of Lemma~\ref{lemma:eigendecomposition}.

Finally, we may recognise that the measurement presented in Lemma~\ref{lemma:attainability}, is precisely the eigendecomposition of the measurement presented in the statement of Thm.~\ref{thm:1}.
\begin{restatable}{lemma}{EigenDecomposition}
 \label{lemma:eigendecomposition}
    The operator  $M_{1\ldots kA}:=\frac{dk}{(k-1+d)}\Big(P^{sym}_{1\ldots k}\otimes \id_A\Big)\left( \id_{1\ldots k}\otimes \ketbra{\phi_d^+}_{kA} \right)\Big(P^{sym}_{1\ldots k}\otimes \id_A\Big)$ is a projector (hence $M\leq \id$), and has its eigendecomposition given by
\begin{align}
    M_{1\ldots kA} = \sum_{i=1}^{k-2+d \choose k-1} \ketbra{r_i}_{1\ldots kA},
\end{align}
where $ \ket{r_i}:=\sqrt{\frac{d}{k(k-1+d)}} \left( \sum_{a=1}^{k} V_{(ak)}\otimes \id_A \Big(\ket{s_i}_{12\ldots(k-1)} \otimes \ket{\phi_d^+}_{kA} \Big)\right)$ are the vectors in Def.~\ref{def:r_i}
\end{restatable}
The proof of Lemma~\ref{lemma:eigendecomposition} will make use of some group representation theory methods and is presented in Sec.~\ref{sec:Proofs}.

\section{Proof of optimality and group representation theory methods} \label{sec:Proofs}
In this section, we present the proof of Thm.~\ref{thm:1}, but before proceeding with the proof, we introduce some ideas borrowed from group representation theory and develop the necessary notation. 

\subsection{Partitions and Young frames}
A partition $\alpha$ of a natural number $k$, denoted by $\alpha \vdash k$, is a sequence of positive numbers $\alpha=(\alpha_1,\alpha_2,\ldots,\alpha_r)$ such that
\begin{align}
\label{eq:partition}
\alpha_1\geq \alpha_2\geq \ldots \geq \alpha_r>0,\qquad \sum_{i=1}^r\alpha_i=k\,.
\end{align}
Every partition can be visualized as a \textit{Young frame} which is a collection of boxes arranged in left-justified rows. The number of Young frames for a fixed number $k$ can be evaluated recursively~\cite{vershik}. In this paper by $\alpha,\mu$ we denote  Young frames with $k,k+1$ boxes respectively. The length of the first column in a given Young frame $\mu$ (equivalently, number of summed elements in~\eqref{eq:partition}) is denoted as a height $\operatorname{ht}(\mu)$.  From the set of all Young frames with $k$ boxes we distinguish a particular one called a symmetric frame denoted by $sym_k$ - $k$ boxes arranged in one row. 
By writing $\mu=\alpha+\Box$ we denote a Young diagram $\mu \vdash k+1$ obtained from $\alpha \vdash k$ by adding one box, while by $\alpha=\mu-\Box$ by subtracting a single box.
\begin{figure}[h!]
\centering
\begin{subfigure}{.5\textwidth}
  \centering
  \includegraphics[width=.6\linewidth]{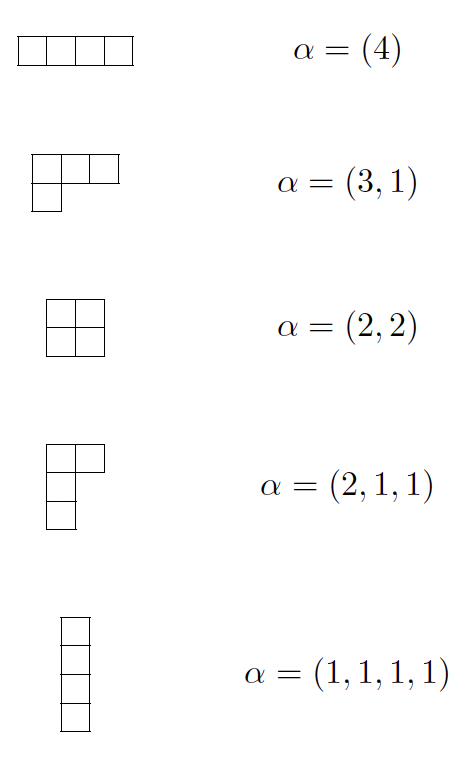}
  \caption*{A}
  \label{fig:young_fames}
\end{subfigure}%
\begin{subfigure}{.5\textwidth}
  \centering
  \includegraphics[width=1\linewidth]{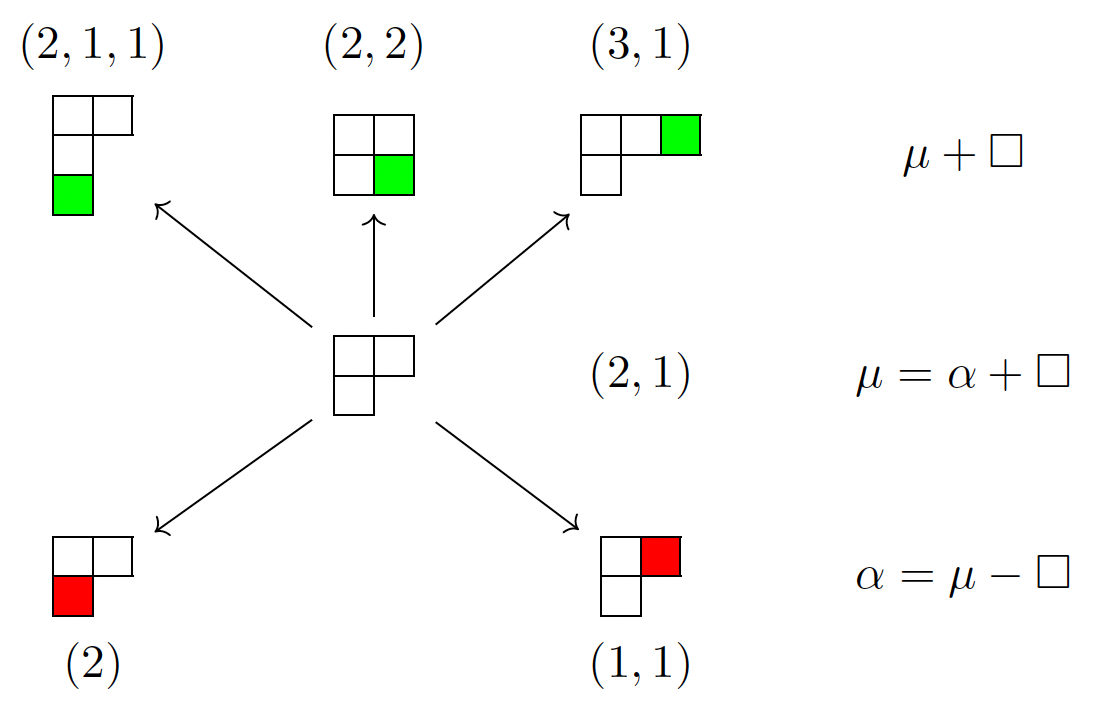}
  \vspace{0.8cm}
  \caption*{B}
  \label{fig:young_relations}
\end{subfigure}
\caption{Panel A presents five possible Young frames for $k=4$. Panel B presents possible Young frames obtained from a frame $\mu = (2,1)$ by adding a single box, depicted here in green, and by subtracting a single box, depicted in red. }
\label{fig:young_frame_and_relations}
\end{figure}

\noindent
With a given Young frame $\mu \vdash k$ we can associate a \textit{Young tableau} which is obtained by filling in the boxes of the Young frame with numbers $1,\ldots,k$. Numbers in the boxes must strictly increase from left to right in every row and from top to bottom in every column. The number of the standard Young tableaux, denoted by $d_\mu$, can be evaluated by using certain combinatorial expressions like the hook-length formula~\cite{ceccherini2010representation,fulton1997young}. With a given Young frame $\mu \vdash k$ we can associate semi-standard Young tableaux. These objects are obtained by filling in the boxes of the Young frame with numbers $1,\ldots, d$, where $d$ is some natural number. In every row, numbers must be arranged in non-decreasing order from the left to the right and strictly increasing order in every column from the top to the bottom. The number of all semi-standard Young frames for a given tableau is denoted by $m_\mu$, and it can also be evaluated by combinatorial rules like the hook-content formula~\cite{ceccherini2010representation,fulton1997young}. It is easy to see that when $d<k$ not all semi-standard Young frames exist. Namely, we have to exclude all Young frames whose height is greater than $d$.
In the particular case, when we consider symmetric frames of $k-1,k,k+1$ boxes,  we know closed expressions for the numbers of the corresponding (semi-)standard Young tableaux $d_{sym_{k-1}}, d_{sym_k},d_{sym_{k+1}}, m_{sym_{k-1},}$ $m_{sym_k}$, and $m_{sym_{k+1}}$:
\begin{align}
&d_{sym_{k-1}}=d_{sym_k}=d_{sym_{k+1}}=1,\\
&m_{sym_{k-1}}={k-2+d\choose k-1} \label{eq:k-1} \\
    &m_{sym_k}={k-1+d\choose k} \label{eq:k}\\ 
    &m_{sym_{k+1}}={k+d\choose k+1}. \label{eq:k+1} 
\end{align}

\subsection{The Schur-Weyl duality}
\label{sec:SW}
We start from defining the \textit{permutational representation} $V:S_k$ $\rightarrow \operatorname{Hom}(\mathcal{(\mathbb{C}}^{d})^{\otimes k})$ of the symmetric group $S_k$ in the Hilbert space $\mathcal{H}=(\mathbb{C}^d)^{\otimes k}$. The elements of $V(S_k)$ permute vectors in $(\mathbb{C}^d)^{\otimes k}$ according to a given permutation $\sigma \in S_k$ as
\begin{align}
\label{eq:actionS_k}
V_\sigma |v_1\>\otimes |v_2\>\otimes
	\ldots \otimes |v_k\>:=|v_{\sigma ^{-1}(1)}\>\otimes |v_{\sigma
			^{-1}(2)}\>\otimes \ldots \otimes |v_{\sigma ^{-1}(k)}\>\,.
\end{align}
The representation $V(S_k)$ extends in a natural way to the
representation of the group algebra $\mathbb{C}[S_k]:=\operatorname{span}_{\mathbb{C}}\{V_\sigma:\sigma \in S_k\}$. In the similar manner, we can define a diagonal action $U^{\otimes k}: \SU(d) \mapsto \operatorname{Hom}(\mathbb{C}^d)^{\otimes { k} }$ of elements $U$ from the special unitary $\SU(d)$. The elements $U^{\otimes k}$ act in  $(\mathbb{C}^d)^{\otimes { k} }$ as
\begin{align}
\label{eq:actionUk}
U^{\otimes k}|v_1\>\otimes |v_2\>\otimes
	\ldots \otimes |v_k\>:=U|v_1\>\otimes U|v_2\>\otimes
	\ldots \otimes U|v_k\>.
\end{align}
It is easy to see that actions Eqs.~\eqref{eq:actionS_k},~\eqref{eq:actionUk} commute and due to this property, there exists a basis in which the tensor product space $(\mathbb{C}^d)^{\otimes k}$, as well as $V_\sigma, U^{\otimes k}$, can be decomposed as 
\begin{align}
\label{wedd}
&(\mathbb{C}^d)^{\otimes k}=\bigoplus_{\substack{\mu \vdash k \\ \operatorname{ht}(\mu)\leq d}} \mathcal{U}_{\mu}\otimes \mathcal{S}_{\mu}\,,\\
&V_\sigma=\bigoplus_{\substack{\mu \vdash k \\ \operatorname{ht}(\mu)\leq d}} \id_{\mathcal{U}_\mu}\otimes \varphi^{\mu}(\sigma)\,,\\
&U^{\otimes k}=\bigoplus_{\substack{\mu \vdash k \\ \operatorname{ht}(\mu)\leq d}} U^{\mu}\otimes \id_{\mathcal{S}_\mu}\,,
\end{align}
where the direct sum runs over all Young frames of $k$ boxes with a height no larger than the local dimension $d$. By $\varphi^{\mu}(\sigma), U^{\mu}$ we denote irreducible representations of $V_\sigma, U^{\otimes k}$, respectively. 
From the decomposition Eq.~\eqref{wedd} we deduce that for a given irreducible representation (irrep) $\mu$ of $S_k$, the space $\mathcal{U}_{\mu}$ is the corresponding multiplicity space of dimension $m_{\mu}$ (multiplicity of irrep $\mu=$ number of semi-standard Young tableaux for integer $d$), while the space $\mathcal{S}_{\mu}$ is {the} representation space of dimension $d_{\mu}$ (dimension of irrep $\mu=$ number of standard Young tableaux). It means permutations are represented non-trivially only on the space $\mathcal{S}_{\mu}$. 

By $P_\mu,P_{\alpha}$ we denote the Young projectors onto irreps of the symmetric groups $S_k,S_{k+1}$ labelled by $\mu$ and $\alpha$ respectively. For a fixed $\mu \vdash k$, the Young projector $P_\mu$ on the space $\mathcal{U}_{\mu}\otimes \mathcal{S}_{\mu}$ is represented on $(\mathbb{C}^d)^{\otimes k}$ as
\begin{align}
\label{Yng_proj}
P_{\mu}=\frac{d_{\mu}}{k!}\sum_{\sigma \in S_k}\chi^{\mu}(\sigma^{-1})V_\sigma,
\end{align}
where $\chi^{\mu}(\sigma^{-1})=\tr(\varphi^{\mu}(\sigma))$ is the irreducible character associated with the irrep indexed by $\mu$.When we have a symmetric Young frame $\mu \equiv sym_k$, we denote the corresponding Young projector by explicitly writing underlying systems as $P^{sym}_{1\ldots k}$.  Then we have
\begin{align}
\label{Yng_proj_sym}
 P^{sym}_{1\ldots k}=\frac{1}{k!}\sum_{\sigma \in S_k}V_\sigma,   
\end{align} 
since all irreducible characters from Eq.~\eqref{Yng_proj} in this case are equal to 1.
For the fixed symmetric group, the Young projectors satisfy the orthogonality property $P_{\mu}P_{\nu}=\delta_{\mu \nu}P_{\mu}$ with the trace rule $\tr(P_{\mu})=m_{\mu}d_{\mu}$. For irrep $\mu$ of $S_k$
we define the natural representation of the matrix basis for the irrep $\mu$ including multiplicities:
\begin{align}
    \label{eq:E}
 \forall \mu \vdash k,\quad i,j=1,\ldots,d_{\mu}\quad   E_{ij}^\mu=
    \frac{d_{\mu}}{k!}\sum_{\sigma\in S_k} \varphi^\mu_{ji}(\sigma^{-1})V_{\sigma}.
\end{align}
Here $\varphi^\mu_{ji}(\sigma^{-1})$ are matrix elements of irrep $\varphi^{\mu}(\sigma^{-1})$ and $i,j=1,\ldots, d_{\mu}$. 
The basis operators given by Eq.~\eqref{eq:E} satisfy the following rules:
\begin{align}
E_{ij}^\mu E_{kl}^\nu=\delta^{\mu \nu}\delta_{jk}E_{il}^{\mu}, \qquad \tr(E_{ij}^\mu)=m_{\mu}\delta_{ij}.
\end{align}
It is easy to see the diagonal operators, i.e. the operators of the form $E^{\mu}_{ij}$, are projectors of rank $m_\mu$.
 There is a connection between Young projectors $P_\mu$ given by Eq.~\eqref{Yng_proj} and irreducible matrix basis operators $E^{\mu}_{ii}$ given by Eq.~\eqref{eq:E} as
\begin{align}
\label{eq:connectionEvP}
\forall \mu\vdash k \quad P_{\mu}=\sum_{i=1}^{d_{\mu}}E^{\mu}_{ii}.
\end{align}
Now consider {the} subgroup $S_{k-1}\subset S_k$. We can restrict an irrep $\mu$ to this subgroup and obtain a representation called the {\it reduced representation}. Then $\alpha$ labels its irreducible components within irrep $\mu$, determining blocks. What is more, each copy of irrep $\alpha$ appears with multiplicity one, so we deal with the multiplicity-free process. We can choose the basis in such a way that
 the matrix elements of $\varphi^\mu(\sigma)$ are written compatible with the $\alpha$ blocks, and we obtain \textit{partially reduced irreducible representation}~\cite{studzinski16}.
Then every label $i$ in Eq.~\eqref{eq:E} can be rewritten by means of a pair $i_\mu=(\alpha, i_\alpha$) -- every index $\alpha$ indicates in which irrep $\alpha$
we are, while the index $i_\alpha$
points us to the position in that irrep.  In the same way, the index $i_\mu$ denotes the position in irrep $\mu$ (the position in the matrix). Using this argumentation, for example, we can write operators $E_{ij}^{\mu}\equiv E_{i_\mu j_\mu}$ given by Eq.~\eqref{eq:E} for $\mu \vdash k$ in terms of the basis adapted to the subgroup $S_{k-1}$ as
\begin{align}
E_{ij}^{\mu}\equiv E_{i_\mu j_\mu}=E^{\alpha \beta}_{i_\alpha j_\beta}(\mu).    
\end{align}

\subsection{The algebra of partially transposed permutation operators}
Before we proceed, let us remind the notion of partial transposition. The \emph{partial transposition} $(\cdot)^{t_k}$ with respect to the $k-$th system
is defined as the linear extension of the ordinary matrix transposition $(\cdot)^t$ with respect to a given basis $\langle i|X^t |j\rangle := \bra{j}X \ket{i}$. Namely, for a bipartite system, the partial transposition with respect to the first subsystem $t_1$ transforms 
\begin{equation}
    t_1:\;|i \rangle\langle j|\otimes |k\rangle\langle l|\mapsto |j\rangle\langle i|\otimes |k\rangle\langle l|\,.
\end{equation}
Analogously, we define the partial transposition with respect to the second subsystem. 
The partial transposition relates the permutation operator $V_{(12)}$ to the Bell state $\ket{\phi^+}=\frac{1}{\sqrt{d}}\sum_{i=1}^d\ket{ii}$ through $V_{(12)}= d  \dyad{\phi^+}^{t_1}$. We can extend the definition of the partial transposition to the multi-partite scenario.

Having the definition of the group algebra $\mathbb{C}[S_k]$ and the partial transposition, we can naturally define the algebra of partially transposed operators with respect to the last subsystem:
\begin{align}
	\mathcal{A}_{k}^{t_{k}}(d):= \operatorname{span}_{\mathbb{C}}\{V^{t_{k}}_\pi:\pi \in S_k\}.
\end{align}
The elements of $\mathcal{A}_{k}^{t_{k}}(d)$ commute with the mixed action of the unitary group of the form $U^{\otimes (k-1)}\otimes \bar{U}$, where the bar denotes complex conjugation, and $U\in U(d)$.
The algebra $\mathcal{A}_{k}^{t_{k}}(d)$ has been extensively studied in the context of the port-based teleportation schemes and their effective quantum circuits, and semidefinite problems for covariant quantum channels, see for example in~\cite{studzinski20,studzinski16,mozrzymas18,Mozrzymas_2018JPA,fei2023,grinko2023gelfand,grinko2024,grinko2023}. Here we only summarize the facts from the point of view of the irreducible representation theory of $\mathcal{A}_{k}^{t_{k}}(d)$ which are necessary for this paper's consistency. For more details and proof, we refer the reader to the cited papers.  
The algebra $\mathcal{A}_{k}^{t_{k}}(d)$ is a direct sum of two ideals (see Proposition 27 in~\cite{mozrzymas2014})
\begin{equation}
\label{A_decomp}
\mathcal{A}_{k}^{t_{k}}(d)=\mathcal{M}\oplus \mathcal{S}=F \ \mathcal{A}_{k}^{t_{k}}(d) \  F\oplus
(id_{\mathcal{A}}-F)\mathcal{A}_{k}^{t_{k}}(d)(id_{\mathcal{A}}-F),
\end{equation}
where the idempotent $F=\sum_{\alpha \vdash k-2}\sum_{\mu=\alpha+\Box}F_{\mu }(\alpha )$ is the identity on the ideal $\mathcal{M}$, i.e., $F=id_{\mathcal{M}}$, and $id_{\mathcal{A}}$ is the identity operator on the whole space. The operators $F_{\mu }(\alpha )$ are projectors onto the irreps of $\mathcal{A}_{k}^{t_{k}}(d)$ contained in the ideal $\mathcal{M}$ and they satisfy the following rules (Theorem 1 in~\cite{studzinski16}):
\begin{align}
F_{\mu}(\alpha)F_{\nu}(\beta)=\delta_{\mu\nu}\delta_{\alpha \beta}F_{\mu}(\alpha),\quad \tr(F_{\mu}(\alpha))=m_{\alpha}d_{\mu}.
\end{align}
Due to Eq.~(145) in paper~\cite{mozrzymas18}, the explicit form of the projectors $F_{\mu }(\alpha )$ in the natural representation is given by
\begin{align}
\label{explicit}
F_{\mu}(\alpha)=\frac{1}{\gamma_{\mu}(\alpha)}P_{\mu}\sum_{a=1}^{k-1}V_{(a,k-1)}P_{\alpha}\otimes V_{(k-1,k)}^{t_k}V_{(a,k-1)}P_\mu\,,
\end{align}
where $P_{\alpha},P_{\mu}$ are the Young projectors onto irreducible spaces labelled by the Young diagrams $\alpha \vdash k-2$ and $\mu \vdash k-1$, respectively, defined in Eq.~\eqref{Yng_proj}, and $\gamma_{\mu}(\alpha)=(k-1)\frac{m_{\mu}d_{\alpha}}{m_{\alpha}d_{\mu}}$. The projectors $F_\mu(\alpha)$ can be understood as an analogue of the Young projectors $P_\mu$ in Eq.~\eqref{Yng_proj}. The full orthonormal irreducible basis of the ideal $\mathcal{M}$ - the analogue of the operators $E^{\mu}_{ij}$ in Eq.~\eqref{eq:E} - when we deal with a single partial transposition, is given by Theorem 11 in~\cite{studzinski20} as
\begin{equation}
F_{i_\mu j_\nu}(\alpha) = \frac{m_\alpha}{\sqrt{m_\mu m_\nu}}E_{i_\mu 1_\alpha}^{ \  \ \alpha}V_{(k-1,k)}^{t_k}E_{1_\alpha j_{\nu}}^{\alpha} \label{eqn:f_operator}.
\end{equation}
Here $1_\alpha$ represents an arbitrarily fixed label since the considered operator does not depend on its choice~\cite{studzinski20} and we have written one of the index the partially reduced notation - see discussion in Section~\ref{sec:SW}. Now, due to Definition 15 presented in~\cite{studzinski20}, the operators $F_{\mu}(\alpha)$ given by Eq.~\eqref{explicit} and operators $F_{i_\mu j_\nu}(\alpha)$ are related by 
\begin{align}
\label{eq:FvsFij}
    F_{\mu}(\alpha)=\sum_{i_\mu}
    F_{i_{\mu} i_{\mu}}(\alpha).
\end{align}
The ideal $\mathcal{S}$ from the direct sum in Eq.~\eqref{A_decomp} is generated by the following set of elements~\cite{mozrzymas2014}:
\begin{equation}
\{V_{\sigma}(\id_{\mathcal{A}}-F) \ : \ \sigma\in S_{k-1}\}.
\end{equation}
In fact,  in the ideal $\mathcal{S}$, only elements of $S_{k-1}$ are represented non-trivially, while elements from $S_k$ for $\sigma(k)\neq k$ are represented by the zero operator. Irreducible representations contained in $\mathcal{S}$ are the standard irreducible representations for $S_{k-1}$ and the Young frames of $k-1$ boxes label them. However, some of them also exist in the ideal $\mathcal{M}$ causing technical difficulties in practical calculations~\cite{mozrzymas2014,Mozrzymas_2018JPA}.  Elements from $S_k$ for $\sigma(k)\neq k$ are fully represented only on the ideal $\mathcal{M}$. 

\subsection{Auxiliary facts regarding the group algebra $\mathbb{C}[S_{k+1}]$ and the algebra $\mathcal{A}_{k}^{t_{k}}(d)$}

 \begin{lemma}
\label{L1}
For the operators $P^{sym}_{1\ldots k+1}$ and $P_\mu$ for $\mu \vdash k$, the following relation holds:
\begin{align}
    P^{sym}_{1\ldots k+1} P_\mu = \delta_{\mu,sym_k}P^{sym}_{1\ldots k+1}.
\end{align}
\end{lemma}
\begin{proof}
The operator $P_\mu$ has the natural representation $P_\mu = \frac{d_\mu}{k!}\sum_{\sigma\in S_k} \chi^\mu(\sigma^{-1})V_{\sigma}$, where $\chi^\mu$ is the irreducible character of the representation $\mu$, and $V_{\sigma}$ is the canonical representation of the permutation $\sigma \in S_k$ on the space $(\mathbb{C}^d)^{\otimes k}$. Then we can write explicitly
\begin{align}
    P^{sym}_{1\ldots k+1} P_\mu &= P^{sym}_{1\ldots k+1} \frac{d_\mu}{k!}\sum_{\sigma\in S_k} \chi^\mu(\sigma^{-1})V_{\sigma} \label{line1} \\
    &= \frac{d_\mu}{k!}\sum_{\sigma\in S_k} \chi^\mu(\sigma^{-1})P^{sym}_{1\ldots k+1} \\
    &=\delta_{\mu,sym_k}P^{sym}_{1\ldots k+1},
\end{align}
where $\frac{1}{k!}\sum_{\sigma\in S_k} \chi^\mu(\sigma^{-1}) = 1 \text{ if } \mu = sym_k \text{, and 0 otherwise}$, and we also have $d_{sym_k}=1$. In Eq.~\eqref{line1}, we use the fact that the symmetric projector $P^{sym}_{1\ldots k+1}$ annihilates all permutation operators $V_{\sigma}$.
\end{proof}

Let us denote by $t_A$ transposition with respect to the system $A$ and consider the symmetric projectors $P^{sym}_{1\ldots k}$ and $P^{sym}_{1\ldots k,A}$. In the next lemma, we find decompositions of $P^{sym}_{1\ldots k} \otimes \id_A$ and $(P^{sym}_{1\ldots k,A})^{t_A}$ in the algebra $\mathcal{A}_{k}^{t_{k}}(d)$.

\begin{lemma}
\label{Lemma2}
Operators $(P^{sym}_{1\ldots k, A})^{t_A}$ and $P^{sym}_{1\ldots k}\otimes \id_A$ have the following decomposition in the algebra $\mathcal{A}_{k}^{t_{k}}(d)$:
\begin{align}
&(P^{sym}_{1,\ldots,k,A})^{t_A}=\frac{d+k}{k}F_{sym_{k}}(sym_{k-1})+\frac{1}{k+1}P^{sym_k}_{\mathcal{S}},\label{Lemmeq:1}\\
&P^{sym}_{1\ldots k}\otimes \id_A =F_{sym_{k}}(sym_{k-1})+P^{sym_k}_{\mathcal{S}}\label{Lemmeq:2}.
\end{align}
The projector $F_{sym_{k}}(sym_{k-1})$ projects on the irrep labelled by the pair $(sym_k,sym_{k-1})$ contained in the ideal $\mathcal{M}$, while the projector $P^{sym_k}_{\mathcal{S}}$ projects on copies of the irrep $sym_k$ contained in the ideal $\mathcal{S}$.
\end{lemma}

\begin{proof}
In the first step, we calculate the overlap of the operator $(P^{sym}_{1,\ldots,k,A})^{t_A}$ with the basis given in Eq.~\eqref{eqn:f_operator}
\begin{align}
 \frac{m_\alpha}{\sqrt{m_{\mu}m_{\nu}}} \tr \Bigl[  (P^{sym}_{1,\ldots,k,A})^{t_A}E_{i_\mu 1_\alpha}^{ \  \ \alpha} V^{t_A}_{(k,A)}E_{1_\alpha i_\nu}^{\alpha} \Bigr] &= \frac{m_\alpha}{\sqrt{m_{\mu}m_{\nu}}} \tr \Bigl[  P^{sym}_{1,\ldots,k,A}E_{i_\mu 1_\alpha}^{ \  \ \alpha} V_{(k,A)}E_{1_\alpha i_\nu}^{\alpha} \Bigr]\label{eq:1}\\
 &= \frac{m_\alpha}{\sqrt{m_{\mu}m_{\nu}}}\tr\Bigl[P^{sym}_{1,\ldots,k,A}V_{(k,A)} E_{1_\alpha i_\nu}^{\alpha}E_{i_\mu 1_\alpha}^{ \  \ \alpha}\Bigr]\label{eq:2}\\
 &=\delta_{\mu\nu}\delta_{i_{\mu}i_{\nu}}\frac{m_{\alpha}}{m_{\mu}}\tr\Bigl[P^{sym}_{1,\ldots,k,A}E^{\alpha \alpha}_{1_\alpha 1_\alpha}(\mu)\Bigr]\label{eq:3}\\
 &=\delta_{\mu\nu}\delta_{i_{\mu}i_{\nu}}\frac{m_{\alpha}}{m_{\mu}d_{\alpha}}\tr\Bigl[P^{sym}_{1,\ldots,k,A}\sum_{d_{\alpha}}E^{\alpha \alpha}_{i_\alpha i_\alpha}(\mu)\Bigr]\label{eq:4}\\
  &=\delta_{\mu\nu}\delta_{i_{\mu}i_{\nu}}\frac{m_{\alpha}}{m_{\mu}d_{\alpha}}\tr\Bigl[P^{sym}_{1,\ldots,k,A}P_\alpha P_\mu\Bigr]\label{eq:5}.
\end{align}
For the equality in line \eqref{eq:1}, we use the property of $\tr(X^{t_A}F_{_{i_\mu j_\nu}}(\alpha)) =\tr(XF_{_{i_\mu j_\nu}}^{t_A}(\alpha))$ and the fact that the operators $E_{i_\mu 1_\alpha}^{ \  \ \alpha}$ and $E_{1_\alpha i_\nu}^{\alpha}$ act trivially on the system $A$. In lines~\eqref{eq:2} and~\eqref{eq:3}, we use the fact that for every $\sigma \in S_{l}$ for $l\leq k+1$, we have $V_{\sigma}P_{1,\ldots,k,A}^{sym}=P_{1,\ldots,k,A}^{sym}V_{\sigma}=P_{1,\ldots,k,A}^{sym}$. Next, in line~\eqref{eq:4}, we use the fact that the trace is constant for every index $1\leq i_\alpha \leq d_{\alpha}$. Finally, in line~\eqref{eq:5}, we use the decomposition of the Young projectors in terms of the operator basis given in Eq.~\eqref{eq:connectionEvP}, and the fact we take only a single copy of the irrep $\alpha$ that is contained in $\mu$. Now, we continue calculations, taking into account that the Young projectors $P_\alpha$ and $P_\mu$ do not act on the system $A$:
\begin{align}
&\frac{m_\alpha}{\sqrt{m_{\mu}m_{\nu}}} \tr \Bigl[  (P^{sym}_{1,\ldots,k,A})^{t_A}E_{i_\mu 1_\alpha}^{ \  \ \alpha} V^{t_A}_{(k,A)}E_{1_\alpha i_\nu}^{\alpha} \Bigr] =\delta_{\mu\nu}\delta_{i_{\mu}i_{\nu}}\frac{m_{\alpha}}{m_{\mu}d_{\alpha}}\tr\Bigl[P^{sym}_{1,\ldots,k,A}P_\alpha P_\mu\Bigr]\\
&=\delta_{\mu\nu}\delta_{i_{\mu}i_{\nu}}\frac{m_{\alpha}}{m_{\mu}d_{\alpha}}\tr\Bigl[P_\alpha P_\mu\tr_A(P^{sym}_{1,\ldots,k,A})\Bigr]\label{eq2:2}\\
&=\delta_{\mu\nu}\delta_{i_{\mu}i_{\nu}}\frac{m_{\alpha}}{m_{\mu}d_{\alpha}}\frac{m_{sym_{k+1}}}{m_{sym_{k}}}\tr\Bigl[P_\alpha P_\mu P^{sym}_{1,\ldots,k}\Bigr]\label{eq2:3}\\
&=\delta_{sym_{k-1}\alpha}\delta_{sym_{k}\mu}\delta_{\mu\nu}\delta_{i_{\mu}i_{\nu}}\frac{m_{sym_{k-1}}}{m_{sym_k}}\frac{m_{sym_{k+1}}}{m_{sym_{k}}}\tr\Bigl[P^{sym}_{1,\ldots,k-1}P^{sym}_{1,\ldots,k}\Bigr]\label{eq2:4}\\
&=\delta_{sym_{k-1}\alpha}\delta_{sym_{k}\mu}\delta_{\mu\nu}\delta_{i_{\mu}i_{\nu}}\frac{m_{sym_{k-1}}}{m_{sym_k}}\frac{m_{sym_{k+1}}}{m_{sym_{k}}}\tr\Bigl[P^{sym}_{1,\ldots,k}\Bigr]\label{eq2:44}\\
&=\delta_{sym_{k-1}\alpha}\delta_{sym_{k}\mu}\delta_{\mu\nu}\delta_{i_{\mu}i_{\nu}}\frac{m_{sym_{k+1}}m_{sym_{k-1}}}{m_{sym_{k}}}\label{eq2:7}.
\end{align}
In line~\eqref{eq2:2}, we exploit Corollary 10 from~\cite{studzinski20}. In line~\eqref{eq2:3}, we use the orthogonality relations $P_{\mu}P_{\nu}=\delta_{\mu \nu}P_{\mu}$ and $d_{sym_{k-1}}=1$, together with the fact that we have $\alpha\in \mu \vdash k$ from the construction, so since $\mu=sym_k$, it must be $\alpha=sym_{k-1}$. In line~\eqref{eq2:4}, we apply Lemma~\ref{L1}, and then in line~\eqref{eq2:44}, the trace rule $\tr(P_{\mu})=m_{\mu}d_{\mu}$ is applied. The final result contained in line~\eqref{eq2:7} shows that the considered trace is non-zero only when $\mu=sym_k$ and $\alpha=sym_{k-1}$ are satisfied, i.e. for $F_{sym_{k}}(sym_{k-1})$, since the sum in~Eq.~\eqref{eq:FvsFij} is trivial in this case. Now we must find the rest of the operator $(P^{sym}_{1,\ldots,k,A})^{t_A}$ which is represented on the ideal $\mathcal{S}$. This operator must be written in the basis of the algebra $\mathcal{A}_{k}^{t_{k}}(d)$. However, we know that on the ideal $\mathcal{M}$, only $F_{sym_{k}}(sym_{k-1})$ survives, so we can assume the following decomposition:
\begin{align}
\label{eq:1kA}
(P^{sym}_{1,\ldots,k,A})^{t_A}=c_1F_{sym_{k}}(sym_{k-1})+c_2P^{sym_k}_{\mathcal{S}},
\end{align}
where  $c_1,c_2\in \mathbb{R}$, and $P^{sym_k}_{\mathcal{S}}=\id_{\mathcal{A}}-F_{sym_{k}}(sym_{k-1})$, since we had to subtract irreps $sym_k$ which are represented on $\mathcal{M}$. In an analogous way, we can determine the following decomposition of the operator $P^{sym}_{1\ldots k}\otimes \id_A$ in the algebra $\mathcal{A}_{k}^{t_{k}}(d)$:
\begin{align}
\label{eq:1k}
P^{sym}_{1\ldots k}\otimes \id_A=F_{sym_{k}}(sym_{k-1})+P^{sym_k}_{\mathcal{S}}.
\end{align}
We have a trivial coefficient in the above decomposition since the left-hand side is a projector, and both $F_{sym_{k}}(sym_{k-1})$ and $P^{sym_k}_{\mathcal{S}}$ are supported on orthogonal spaces, so $\tr(F_{sym_{k}}(sym_{k-1}) P^{sym_k}_{\mathcal{S}})=0$ has to hold. This already proves the expression given by Eq.~\eqref{Lemmeq:2}, since as we will show later on, we are not interested in the explicit form of the projector $P^{sym_k}_{\mathcal{S}}$. 

Now we are in the position to determine the unknown coefficients in Eq.~\eqref{eq:1kA}. Let us calculate the overlap of Eq.~\eqref{eq:1kA} with $F_{sym_k}(sym_{k-1})$, so that
 \begin{align}
     \tr\big((P^{sym}_{1,\ldots,k,A})^{t_A} F_{sym_k}(sym_{k-1}) \big) = c_1 m_{sym_{k-1}},
 \end{align}
 where from previous calculations, i.e. from Eq.~\eqref{eq2:7}, we know that the left-hand side is equal to $\frac{m_{sym_{k+1}}m_{sym_{k-1}}}{m_{sym_{k}}}$. This gives the value of the coefficient $c_1$ as
 \begin{align}
 \label{eq:c1}
c_1=\frac{m_{sym_{k+1}}}{m_{sym_{k}}}=\frac{d+k}{k+1}.
 \end{align}
To evaluate the coefficient $c_2$, let us combine the expression given by Eq.~\eqref{eq:1kA} with Eq.~\eqref{eq:1k} to remove the projector $P^{sym_k}_{\mathcal{S}}$ as
\begin{align}
(P^{sym}_{1,\ldots,k,A})^{t_A}=\frac{m_{sym_{k+1}}}{m_{sym_{k}}}F_{sym_k}(sym_{k-1})+c_2\left(P^{sym}_{1\ldots k}\otimes \id_A-F_{sym_k}(sym_{k-1})\right).
\end{align}
Evaluating the trace from both sides and taking into account that $\tr(F_{sym_k}(sym_{k-1}))=m_{sym_{k-1}}$ (see Theorem 1 in~\cite{studzinski16}), we obtain
\begin{align}
m_{sym_{k+1}}=\frac{m_{sym_{k+1}}m_{sym_{k-1}}}{m_{sym_{k}}}+dc_2m_{sym_k}-c_2m_{sym_{k-1}}.
\end{align}
This finally gives
\begin{align}
\label{eq:c2}
c_2=\frac{m_{sym_{k+1}}}{m_{sym_k}}\frac{\left(1-\frac{m_{sym_{k-1}}}{m_{sym_k}}\right)}{\left(d-\frac{m_{sym_{k-1}}}{m_{sym_k}}\right)}=\frac{d+k}{k+1}\left(\frac{1-\frac{d+k-1}{k}}{d-\frac{d+k-1}{k}}\right)=\frac{1}{k+1}.
\end{align}
In Eqs.~\eqref{eq:c1} and \eqref{eq:c2}, we have explicit formulas for the coefficients $c_1,c_2$, so we proved the decomposition given by Eq.~\eqref{Lemmeq:1}. This finishes the proof.
\end{proof}

\subsection{The proof of the main theorem}
Now we are in a position to formulate and prove the main result of this section.
\MainThm* 
\begin{proof}
The proof consists of two main steps. In the first step, we show that the problem of finding the optimal success probability $p(d,k)$ can be recast as an SDP problem. In the second step, by exploiting the internal symmetries of the problem, we show the achievability of Eq.~\eqref{eq:achev} by explicitly constructing a feasible solution to SDP, i.e. postulating a particular form of the measurements given by Eq.~\eqref{eq:achev2}.

From the general consideration presented in Sec.~\ref{Section2}, the probabilistic teleportation condition can be written as
\begin{align}
	\tr_{1\ldots k,A}\Big(\ketbra{\psi}{\psi}^{\otimes k}_{12\ldots k} \otimes \ketbra{\phi^+_d}{\phi^+_d}_{AB} \;  M_{1\ldots kA}\otimes \id_{B}\Big)= p(d,k) \ketbra{\psi}{\psi}_B.
\end{align}
This condition is equivalent to the following two relations:
\begin{align}
    \tr\Big(\ketbra{\psi}{\psi}^{\otimes k}_{12\ldots k} \otimes \ketbra{\phi^+_d}{\phi^+_d}_{AB} \; M_{1\ldots kA}\otimes \id_{B} \Big)&= p(d,k), \label{eq:LHS1}\\
    \bra{\psi}\tr_{1\cdots k,A}\Big(\ketbra{\psi}{\psi}^{\otimes k}_{12\ldots k} \otimes \ketbra{\phi^+_d}{\phi^+_d}_{AB} \; M_{1\ldots kA}\otimes \id_{B}\Big)\ket{\psi}&= p(d,k).\label{eq:LHS2}
\end{align}
Before proceeding, let us notice that the operator $M_{1\ldots kA}$ may be assumed to respect several symmetries. In particular, the following commutation relations:
\begin{align}
&	[M_{1\ldots kA}, U^{\otimes k}\otimes \overline{U}]=0, \quad \forall U\in\SU(d)  \label{r1}\\
&	[M_{1\ldots kA},  V_\sigma \otimes \id_A]=0, \quad \forall \sigma\in S_k \label{r2},
\end{align}
must hold, where $V_\sigma$ is the canonical representation of the permutation $\sigma \in S_k$ on the space $(\mathbb{C}^d)^{\otimes k}$ {defined in Eq.~\eqref{eq:actionS_k}}, the identity operator $\id_A$ acts on the system $A$, and the bar denotes complex conjugation.
From the unitary group symmetry in Eq.~\eqref{r1}, we can show that the left-hand sides of Eqs.~\eqref{A_decomp} and \eqref{eq:LHS2} do not depend on the input state $\ket{\psi}$.  Therefore, we can rewrite the above two conditions as
\begin{align}
    p(d,k)&= \int \mathrm{d} \psi \tr\Big(\ketbra{\psi}{\psi}^{\otimes k}_{12\ldots k} \otimes \ketbra{\phi^+_d}{\phi^+_d}_{AB} \;  M_{1\ldots kA}\otimes \id_{B}\Big),\label{c1}\\
    p(d,k)&= \int \mathrm{d} \psi \bra{\psi}\tr_{1\cdots kA}\Big(\ketbra{\psi}{\psi}^{\otimes k}_{12\ldots k} \otimes \ketbra{\phi^+_d}{\phi^+_d}_{AB} \; M_{1\ldots kA}\otimes \id_{B}\Big)\ket{\psi} \label{eqn:telep_1}{,}
\end{align}
where $\dd \psi$ is the uniform distribution given by $\ket{\psi}=U\ket{\psi_0}$ for a Haar-random unitary and a fixed state $\ket{\psi_0}$.
The first condition~\eqref{c1} can be further evaluated as
\begin{align}
    p(d,k)
    &= \int \mathrm{d} \psi \tr\Big(\ketbra{\psi}{\psi}^{\otimes k}_{12\ldots k}\otimes \ketbra{\phi^+_d}{\phi^+_d}_{AB} \;  M_{1\ldots kA}\otimes \id_{B}\Big)\\
    &= \frac{1}{d} \int \mathrm{d} \psi \tr\Big(\ketbra{\psi}{\psi}^{\otimes k}_{12\ldots k}\otimes \id_{A} \; M_{1\ldots kA}\Big)\\
    &= \frac{1}{d} \tr\Big(\frac{P^{sym}_{1\ldots k}}{m_{sym_k}}\otimes \id_{A} \; M_{1\ldots kA}\Big), \label{eq:nice_formula}
\end{align}
where $m_{sym_k}$ is the multiplicity of the projector $P^{sym}_{1\ldots k}$.
The second condition~\eqref{eqn:telep_1} can be rewritten as
\begin{align}
    p(d,k)
    &= \int \mathrm{d} \psi \bra{\psi}\tr_{1\cdots kA}\Big(\ketbra{\psi}{\psi}^{\otimes k}_{12\ldots k}\otimes \ketbra{\phi^+_d}{\phi^+_d}_{AB} \; M_{1\ldots kA}\otimes \id_{B}\Big)\ket{\psi}\label{eq:max}\\
    &= \int \mathrm{d} \psi \tr \Big(\ketbra{\phi^+_d}{\phi^+_d}_{AB}(\id_{A} \otimes \ketbra{\psi}{\psi}_{B}) \otimes \ketbra{\psi}{\psi}^{\otimes k}_{12\ldots k} \;  M_{1\ldots kA}\otimes \id_{B}\Big)\\
    &= \int \mathrm{d} \psi \tr \Big(\ketbra{\phi^+_d}{\phi^+_d}_{AB} (\ketbra{\psi}{\psi}^t_{A} \otimes \id_{B}  ) \otimes \ketbra{\psi}{\psi}^{\otimes k}_{12\ldots k} \;  M_{1\ldots kA}\otimes \id_{B}\Big)\\
    &= \frac{1}{d}\int \mathrm{d} \psi \tr \Big(\ketbra{\psi}{\psi}^{\otimes k}_{12\ldots k}\otimes \ketbra{\psi}{\psi}^t_{A} \; M_{1\ldots kA}\Big)\\
    &= \frac{1}{d} \tr \Big[\frac{(P^{sym}_{1\ldots k,A})^{t_A}}{m_{sym_k}} \; M_{1\ldots kA}\Big],\label{eqn:telep_2}
\end{align}
where $t_A$ is the partial transpose with respect to the system $A$ and $m_{sym_{k+1}}$ is the multiplicity of the symmetric projector $P^{sym}_{1\ldots k,A}$ acting on systems $1,\ldots, k,A$. Now, collecting our considerations, we see that the problem of finding the optimal probability of success $p(d,k)$ can be recast as an SDP problem:
\begin{align}
	p(d,k)=\max_{M_{1\ldots kA}\in(\mathbb{C}^d)^{\otimes (k+1)}} & \frac{1}{d} \tr\Big(\frac{P^{sym}_{1\ldots k}}{m_{sym_k}} \otimes \id_{A} \; M_{1\ldots kA}\Big), \label{eq:opt1} \\
	\text{such that: }&\\
&\tr\Big(\frac{P^{sym}_{1\ldots k}}{m_{sym_k}}\otimes \id_{A} \; M_{1\ldots kA}\Big)=\tr \Big[\frac{(P^{sym}_{1\ldots k,A})^{t_A}}{m_{sym_{k+1}}} \; M_{1\ldots kA}\Big],\label{r01} \\
&[M_{1\ldots kA}, U^{\otimes k}\otimes \overline{U}]=0, \quad \forall U\in\SU(d)  \label{r11}\\
&[M_{1\ldots kA}, V_\pi \otimes \id_A]=0, \quad \forall \pi\in S_k \label{r21},\\
&0\leq M_{1\ldots kA}\leq \id. \label{eq:opt4} 
\end{align}
The maximisation in line~\eqref{eq:opt1} is over all operators acting on the space $(\mathbb{C}^d)^{\otimes (k+1)}$ with additional constraints shown in lines~\eqref{r01}-\eqref{eq:opt4}. 
Now we are going to solve the above optimisation problem using symmetries appearing in its formulation. The commutation relation given by Eq.~\eqref{r01} implies that the measurement $M_{1\ldots kA}$ belongs to the algebra of the partially transposed permutations $\mathcal{A}_{k+1}^{t_A}(d)$ where the transposition $t_A$ is applied with respect to the $A^{\text{th}}$ system. The second commutation relation given by Eq.~\eqref{r11} implies that the operator $M_{1\ldots kA}$ must be constant on the irreps of $S_k$. Due to our discussion,  we can assume that $M_{1\ldots kA}$ is a linear combination of projectors $F_\mu(\alpha)$ living on the ideal $\mathcal{M}$ and projectors on irreps $\mu\vdash k$ contained in the ideal $\mathcal{S}$:
\begin{equation}
\label{eqn:decomp1}
M_{1\ldots kA}=\sum_{\alpha}\sum_{\mu \in \alpha}a_{\alpha}(\mu)F_\mu(\alpha)+\sum_{\mu}a_{\mu}P^{\mu}_{\mathcal{S}},\qquad \forall \alpha \vdash k-1, \forall \mu \vdash k \quad a_{\alpha}(\mu),a_\mu\geq 0.
\end{equation}
The projectors $F_{\mu}(\alpha)$ for $\mu\in\alpha$ project on the irreducible components of the algebra $\mathcal{A}_{k+1}^{t_A}(d)$, and they have been studied in the context of the port-based teleportation~\cite{studzinski16,mozrzymas18}. Moreover, from Theorem 1 in~\cite{studzinski16}, we have $[F_{\mu}(\alpha),P_{\nu}]=0$ for $\nu\vdash k$, so indeed the measurement is constant on the irreps of $S_k$.  We can simplify the form of the measurement represented by Eq.~\eqref{eqn:decomp1} more. Namely, we know that {the measurement} must satisfy the relation given by Eq.~\eqref{r01}. On the other hand, thanks to Lemma~\ref{Lemma2}, we know the decomposition of the operators $P^{sym}_{1\ldots k}\otimes \id_{A}$, $(P^{sym}_{1\ldots k,A})^{t_A}$ in the algebra $\mathcal{A}_{k+1}^{t_A}(d)$. Combining these facts, we can write $M_{1\ldots kA}$ given by Eq.~\eqref{eqn:decomp1} without loss of generality as
 \begin{align}
 \label{eq:Mreduced}
     M_{1,\ldots,k,A} = a_{sym_k}(sym_{k-1})F_{sym_k}(sym_{k-1})+a_{sym_k}P_{\mathcal{S}}^{sym},
 \end{align}
for $0\leq a_{sym_k}(sym_{k-1}),a_{sym_k}\leq 1$. Then, starting from the right-hand side, we can rewrite Eq.~\eqref{r01} in terms of the result of Lemma~\ref{Lemma2} and Eq.~\eqref{eq:Mreduced} as
 \begin{align}
     &\frac{1}{m_{sym_{k+1}}}\tr\bigl[(P^{sym}_{1,\ldots,k,A})^{t_A} M_{1,\ldots,k,A} \bigr] =\\
     &=\frac{1}{m_{sym_{k+1}}}\tr\bigl[\big(\frac{d+k}{k+1}F_{sym_k}(sym_{k-1})+\frac{1}{k+1}P_{\mathcal{S}}^{sym}\big) \big(a_{sym_k}(sym_{k-1})F_{sym_k}(sym_{k-1})+a_{sym_k}P_{\mathcal{S}}^{sym}\big)\bigr]\\
     &= \frac{1}{m_{sym_{k+1}}} \Big(\frac{d+k}{k+1}a_{sym_k}(sym_{k-1})m_{sym_{k-1}}+\frac{a_{sym_k}}{k+1}\tr (P_{\mathcal{S}}^{sym})\Big)\\
     &= \frac{k}{k-1+d}a_{sym_k}(sym_{k-1})+\frac{1}{k+1}\frac{a_{sym_k}}{m_{sym_{k+1}}}\tr (P_{\mathcal{S}}^{sym}), \label{eqn:pright}
 \end{align}
since $\tr(F_{sym_k}(sym_{k-1})) = m_{sym_{k-1}}$. Now, applying Eq.~\eqref{eq:Mreduced} and Lemma~\ref{Lemma2} to the left-hand side of Eq.~\eqref{r01}, we obtain
\begin{align}
    &\frac{1}{m_{sym_{k}}}\tr\bigl[P^{sym}_{1,\ldots,k}\otimes \id_A M_{1,\ldots,k,A} \bigr]=\\
    &=\frac{1}{m_{sym_{k}}}\tr\bigl[\big(F_{sym_k}(sym_{k-1}) + P_{\mathcal{S}}^{sym} \big) \big(a_{sym_k}(sym_{k-1})F_{sym_k}(sym_{k-1})+a_{sym_k}P_{\mathcal{S}}^{sym}\big) \bigr]\\
    &= \frac{1}{m_{sym_{k}}}\Big(a_{sym_k}(sym_{k-1}) m_{sym_{k-1}} + a_{sym_k} \tr( P_{\mathcal{S}}^{sym})\Big)\\
    &= a_{sym_k}(sym_{k-1})\frac{m_{sym_{k-1}}}{m_{sym_{k}}}   +  \frac{a_{sym_k}}{m_{sym_{k}}}\tr( P_{\mathcal{S}}^{sym})\\
    &=\frac{k}{k-1+d}a_{sym_k}(sym_{k-1})+\frac{a_{sym_k}}{m_{sym_{k}}}\tr( P_{\mathcal{S}}^{sym}).\label{eqn:pleft}
\end{align}
 Comparing lines \eqref{eqn:pright} and \eqref{eqn:pleft}, we see that $a_2=0$ so that $a_1$ must be equal to $1$ because of Eq.~\eqref{eq:max}, thus the form of the measurement is 
 \begin{align}
     M_{1,\ldots,k,A} = F_{sym_k}(sym_{k-1}).
 \end{align}
First, let us recall the explicit form of the projectors $F_\mu(\alpha)$ from Eq.~(145) in paper~\cite{mozrzymas18}. Using our indexation of systems adopted for this paper, the projectors take the following form:
\begin{align}
\label{eq:Fproj}
    F_\mu(\alpha) = \frac{1}{\gamma_\mu(\alpha)} P_\mu \sum_{a=1}^k V_{(k,a)} P_\alpha \otimes V^{t_A}_{(k,A)} V_{(k,a)}P_\mu,\qquad  \gamma_\mu(\alpha) = k\frac{m_\mu d_\alpha}{m_\alpha d_\mu},
\end{align}
where $t_A$ denotes the partial transposition with respect to the $A^{\text{th}}$ subsystem (Alice), $V_{(k,a)}$ is a permutation between systems $k$ and $a\in\{1,2,\ldots,k\}$, and $V_{(k,A)}$ is a permutation between system $k$ and $A$. The Young projector $P_{\alpha}$ acts on $k$ systems except $a$, however, for simplification, we use a single symbol to denote all of them. In our case $\mu = sym_k$, $\alpha = sym_{k-1}$ and $d_{sym_{k-1}} = d_{sym_k} = 1$ thus $\gamma_{sym_k}(sym_{k-1}) = k\frac{m_{sym_k}}{m_{sym_{k-1}}}=\frac{1}{k-1+d}$. 

In the second part of the proof, we are going to prove the explicit form of the measurement given by Eq.~\eqref{eq:achev2}. Up to now, we have shown that Alice's measurement is of the form $M_{1\ldots kA}=F_{sym_k}(sym_{k-1})$. Thus, we obtain
\begin{align}
M_{1\ldots kA}&=\frac{1}{k-1+d}\sum_{a=1}^k V_{(k,a)}P^{sym}_{1\ldots k}(P^{sym}_{1\ldots k-1}\otimes V^{t_A}_{(k,A)}) V_{(k,a)}P^{sym}_{1\ldots k}\\
&=\frac{dk}{k-1+d}P^{sym}_{1\ldots k} (P^{sym}_{1\ldots k-1}\otimes P^+_{kA})P^{sym}_{1\ldots k}\label{eqq:mes_Simple}\\
&=\frac{dk}{k-1+d}P^{sym}_{1\ldots k} (\id_{1\ldots(k-1)}\otimes P^+_{kA})P^{sym}_{1\ldots k}, \label{eqq:mes_simple}\\
&=\frac{dk}{(k-1+d)}\Big(P^{sym}_{1\ldots k}\otimes \id_A\Big)\Big(\id_{1\ldots(k-1)}\otimes \ketbra{\phi_d^+}_{kA} \Big)\Big(P^{sym}_{1\ldots k}\otimes \id_A\Big),\label{eqq:mes_simplE}
\end{align}
where we used the fact that for every $1\leq a\leq k $ we have $V_{(ka)}P^{sym}_{1\ldots k}=P^{sym}_{1\ldots k}V_{(ka)}=P^{sym}_{1\ldots k}$, relation $V^{t_A}_{(k,A)}=dP^+_{kA}=d\ketbra{\phi_d^+}_{kA}$, and finally in line~\eqref{eqq:mes_Simple}, we exploited Lemma~\ref{L1}.

In the last step of the proof, we calculate the resulting success probability, knowing the fact that maximum in Eq.~\eqref{eq:opt1} is achieved by measurement of the form given in Eq.~\eqref{eqq:mes_simplE}. We have
\begin{align}
    p(d,k)&= \frac{1}{d} \tr\Big(\frac{P^{sym}_{1\ldots k}}{m_{sym_k}} \otimes \id_{A} \; M_{1\ldots kA}\Big)=\frac{dk}{k-1+d}\frac{1}{m_{sym_k}} \tr\Big(P^{sym}_{1\ldots k}\otimes \id_{A} \;  P^+_{kA}\Big)\\
   &=\frac{dk}{k-1+d}\frac{1}{m_{sym_k}} \tr\Big(P^{sym}_{1\ldots k} \;  \tr_A(P^+_{kA})\Big)=\frac{k}{d(k-1+d)}\frac{1}{m_{sym_k}} \tr\Big(P^{sym}_{1\ldots k}\Big)\\
    &=\frac{k}{d(k-1+d)},
\end{align}
since $\tr_A(P^+_{kA})=\id_a/d$, and $\tr(P^{sym}_{1\ldots k})=m_{sym_k}$.
Finally, we obtain the expression given by Eq.~\eqref{eq:achev}, thus the proof is completed.
\end{proof}

\subsection{Proof of Lemma~\ref{lemma:eigendecomposition} and the eigendecomposition of $M$}
Now we re-state Lemma~\ref{lemma:eigendecomposition} and present its proof.
\EigenDecomposition* 
\begin{proof}
Let us rewrite the expression given in line~\eqref{eq:Fproj} for $\mu=sym_k$ and $\alpha=sym_{k-1}$ as
\begin{align}
\label{eq:mes_simple}
M_{1\ldots kA}=F_{sym_k}(sym_{k-1})&=\frac{1}{k-1+d}P^{sym}_{1\ldots k}\sum_{a=1}^k V_{(k,a)} P^{sym}_{1\ldots k-1} \otimes V^{t_A}_{(k,A)} V_{(k,a)}P^{sym}_{1\ldots k}\\
&=\frac{dk}{k-1+d}P^{sym}_{1\ldots k}(P^{sym}_{1\ldots k-1} \otimes P^+_{kA}) P^{sym}_{1\ldots k},\label{eq:mes_simplea3}
\end{align}
where we used the fact that for every $1\leq a\leq k $, we have $V_{(ka)}P^{sym}_{1\ldots k}=P^{sym}_{1\ldots k}V_{(ka)}=P^{sym}_{1\ldots k}$ and the relation $V^{t_A}_{(k,A)}=dP^+_{kA}$.
Using the definition of the Young projectors $P^{sym}_{1\ldots k}$ given by Eq.~\eqref{Yng_proj_sym} and the fact that we can generate all elements $\sigma \in S_k$ using only elements of the left coset $\Sigma_k:=\{V_{(ak)} \ |  1\leq a\leq k\}$ and elements from the group $S_{k-1}$ by writing $V_\sigma=V_{(ak)}V_\pi$ for some $\pi \in S_{k-1}$ and some $1\leq a\leq k$~\footnote{The same decomposition holds also for multiplication from the right-hand side, i.e. we have $V_\sigma=V_{\pi'}V_{(ak)}$ for some $\pi'\in S_{k-1}$ and some $1\leq a\leq k$.}, we can continue rewriting Eq.~\eqref{eq:mes_simplea3} as
\begin{align}
\label{eq:mes_simple2}
M_{1\ldots kA}&=\frac{dk}{k-1+d}\frac{1}{(k!)^2}\sum_{a=1}^k\sum_{\pi\in S_{k-1}}V_{(ak)}V_\pi (P^{sym}_{1\ldots k-1} \otimes P^+_{kA})\sum_{a'=1}^k\sum_{\pi'\in S_{k-1}}V_{\pi'}V_{(a'k)}\\
&=\frac{dk}{k-1+d}\frac{1}{(k!)^2}\sum_{a,a'=1}^kV_{(ak)}\left(V_\pi P^{sym}_{1\ldots k-1}V_{\pi'}\otimes  P^+_{kA}\right)V_{(a'k)}\label{eq:mes_simple33}\\
&=\frac{dk}{k-1+d}\frac{(k-1)!}{(k!)^2}\sum_{a,a'=1}^kV_{(ak)}\left(P^{sym}_{1\ldots k-1}\otimes  P^+_{kA}\right)V_{(a'k)}\\
&=\frac{d}{k(k-1+d)}\sum_{a,a'=1}^kV_{(ak)}\left(P^{sym}_{1\ldots k-1}\otimes  P^+_{kA}\right)V_{(a'k)}.\label{eq:mes_simple3}
\end{align}
In line~\eqref{eq:mes_simple33} we used property that $V_{\pi}P^{sym}_{1\ldots k-1}=P^{sym}_{1\ldots k}V_{\pi}=P^{sym}_{1\ldots k-1}$  for every $\pi \in S_{k-1}$. Now, let us notice that Young projector $P^{sym}_{1\ldots k-1}$ can be written in terms of vectors $\{\ket{s_i}_{1\ldots k-1}\}_{i=1}^{k-2+d \choose k-1}$ from Definition~\ref{def:r_i} as $P^{sym}_{1\ldots k-1}=\sum_i |s_i\>\<s_i|_{1\ldots k-1}$, thus together with $P^+_{kA}=|\phi_d^+\>\<\phi_d^+|_{kA}$, we rewrite Eq.~\eqref{eq:mes_simple3} as 
\begin{align}
&M_{1\ldots kA}=\sum_{i=1}^{k-2+d \choose k-1}\left(\sqrt{\frac{d}{k(k-1+d)}}\sum_{a=1}^kV_{(ak)}|s_i\>_{1\ldots k-1}\otimes |\phi_d^+\>_{kA}\right)\left(\sqrt{\frac{d}{k(k-1+d)}}\sum_{a'=1}^k\<s_i|_{1\ldots k-1}\otimes \<\phi^+_d|_{kA} V_{(a'k)}\right).
\end{align}
Vectors in the brackets are precisely the vectors from Definition~\ref{def:r_i}. Thus, we proved that $M_{1\ldots kA}=\sum_i|r_i\>\<r_i|$ and obtain the expression given by Eq.~\eqref{eq:achev2}. 
\end{proof}

\subsection{Alternative proof to Lemma~\ref{lemma:attainability}} \label{sec:alternative}
We now present an alternative proof to Lemma~\ref{lemma:attainability} which will make use of Lemma~\ref{lemma:eigendecomposition}.
\Attainability* 
\begin{proof}
    We start the proof by showing that $\tr_{1\ldots kA}\Big(\ketbra{\psi}{\psi}^{\otimes k}_{12\ldots k}\otimes \ketbra{\phi^+_d}{\phi^+_d}_{AB} \; M_{1\ldots kA}\otimes \id_{B}\Big)$ is proportional to $\ketbra{\psi}$. For that, we recall that $M_{1\ldots kA}:=\frac{d}{(k-1+d)}\Big(P^{sym}_{1\ldots k}\otimes \id_A\Big)\left(\sum_{a=1}^k \id_{\overline{a}}\otimes \ketbra{\phi_d^+}_{aA} \right)$. We start by analysing the quantity
\begin{align}
    &\tr_{1\ldots kA}\Big(\ketbra{\psi}{\psi}^{\otimes k}_{12\ldots k}\otimes \ketbra{\phi^+_d}{\phi^+_d}_{AB} \; \left(\Big(P^{sym}_{1\ldots k}\otimes \id_A\Big)\Big(\id_{1\ldots(k-1)}\otimes \ketbra{\phi_d^+}_{kA} \Big)\right)\otimes \id_{B}\Big), \\
    =&\tr_{1\ldots kA}\Big(\ketbra{\psi}{\psi}^{\otimes k}_{12\ldots k}\otimes \ketbra{\phi^+_d}{\phi^+_d}_{AB} \;\Big(\id_{1\ldots(k-1)}\otimes \ketbra{\phi_d^+}_{kA} \otimes \id_{B}\Big)\Big),\\
    =&\tr_{kA}\Big(\ketbra{\psi}_k\otimes \ketbra{\phi^+_d}{\phi^+_d}_{AB} \;\Big(\ketbra{\phi_d^+}_{kA} \otimes \id_{B}\Big)\Big),\\
    =&\tr_{kA}\Big(\id_k\otimes \ketbra{\phi^+_d}{\phi^+_d}_{AB} \;\Big(\ketbra{\psi}_k\otimes \id_A \ketbra{\phi_d^+}_{kA} \otimes \id_{B}\Big)\Big),\\
    =&\tr_{kA}\Big(\id_k\otimes \ketbra{\phi^+_d}{\phi^+_d}_{AB} \;\Big(\id_k\otimes \ketbra{\psi}^T_A \ketbra{\phi_d^+}_{kA} \otimes \id_{B}\Big)\Big),\\
    =&\tr_{kA}\Big(\id_k\otimes \ketbra{\phi^+_d}{\phi^+_d}_{AB} \Big(\id_k\otimes \ketbra{\psi}^T_A \otimes\id_B\Big) \;\Big(\ketbra{\phi_d^+}_{kA} \otimes \id_{B}\Big)\Big),\\
    =&\tr_{kA}\Big(\id_k\otimes \ketbra{\phi^+_d}{\phi^+_d}_{AB} \Big(\id_k\otimes \id_A \otimes\ketbra{\psi}_B\Big) \;\Big(\ketbra{\phi_d^+}_{kA} \otimes \id_{B}\Big)\Big),\\
    =&\tr_{A}\Big(\ketbra{\phi^+_d}{\phi^+_d}_{AB} \Big(\id_A \otimes\ketbra{\psi}_B\Big) \;\Big(\frac{\id_A}{d} \otimes \id_{B}\Big)\Big),\\
    =&\frac{\ketbra{\psi}_B}{d^2}. 
\end{align}
By analogous calculations, we see that $\tr_{1\ldots kA}\Big(\ketbra{\psi}{\psi}^{\otimes k}_{12\ldots k}\otimes \ketbra{\phi^+_d}{\phi^+_d}_{AB} \; M_{1\ldots kA}\otimes \id_{B}\Big)$ is proportional to $\ketbra{\psi}$.

Now, in order to evaluate the quantity $p(d,k)$, it is enough to take the trace on the system in $\H_B$ on both sides of Eq.~\eqref{eq:psLemma} to find
\begin{align}
    p(d,k) =& \tr_{1\ldots kAB}\Big(\ketbra{\psi}{\psi}^{\otimes k}_{12\ldots k}\otimes \ketbra{\phi^+_d}{\phi^+_d}_{AB} \; M_{1\ldots kA}\otimes \id_{B}\Big) \\
     =& \frac{d}{(k-1+d)} \sum_{a=1}^k \tr\Big(\ketbra{\psi}{\psi}^{\otimes k}_{12\ldots k}\otimes \ketbra{\phi^+_d}{\phi^+_d}_{AB} \; \Big(P^{sym}_{1\ldots k}\otimes \id_A\Big)\left( \id_{\overline{a}}\otimes \ketbra{\phi_d^+}_{aA} \right)\otimes \id_{B}\Big) \\
      =&\frac{d}{(k-1+d)} \sum_{a=1}^k \Big(\frac{1}{d^2}\Big) \\
      =&\frac{d}{(k-1+d)} \frac{k}{d^2} \\
      =&\frac{k}{d(k-1+d)},
\end{align}
which concludes the proof.
\end{proof}

\section{Discussions and future directions}
In this work, we have introduced the concept of multicopy state teleportation and showed how the optimal performance in this task is obtained.
We have also connected the problem of multicopy state teleportation with the problem of storing a quantum program, i.e., an arbitrary quantum channel, into a quantum state, and then universally retrieving the stored quantum program and applying it on an input state in a scenario where $k$ copies of the input state are available.

It is also useful to place our results in the broader context of programmable quantum processing protocols optimised for different resources. In particular, our setting is complementary to the representation-matching framework introduced by Yang and Hayashi~\cite{PRXQuantum.2.020327}, where the main emphasis is on communication and memory costs rather than on maximising the success provability of exact teleportation.

More precisely, the two approaches address distinct operational regimes. In our work, the primary optimisation goal is the success probability of exact, correction-free teleportation when multiple copies of the input state are available. By contrast, the representation matching approach is designed to reduce communication and memory overhead by exploiting representation theoretic encodings of the relevant quantum information. Thus, while both frameworks exploit symmetry and representation-theoretic structure, they do so for different resource objectives.

In the representation-matching framework, the main optimisation target is communication or memory cost. For storage-and-retrieval of gate arrays, Yang and Hayashi show that a probabilistic protocol can achieve the optimal logarithmic memory scaling, while exact deterministic approaches based on directly storing the relevant operation would in general require substantially larger resources.  In contrast, our protocol operates in a different regime: rather than optimising memory or communication, it uses multiple copies of the input state and achieves an exact teleportation primitive without any correction step, whose probability approaching $1/d$ for large $k$. 
Our protocol also has a very simple heralding structure: Alice only needs to communicate a one-bit success/failure flag. This should not be viewed as a direct comparison with the Yang–Hayashi framework, whose primary cost measures are quantum communication and quantum memory.
Therefore, the two approaches should be viewed as complementary rather than competing. In Table~\ref{tab:comparison_rm} we compare the most important features between the two protocols discussed here.

\begin{table}[t]
\centering
\renewcommand{\arraystretch}{1.2}
\begin{tabular}{|p{3.1cm}|p{5.0cm}|p{5.0cm}|}
\hline
\textbf{Feature} & \textbf{This work} & \textbf{Representation matching} \\
\hline
Primary optimisation goal
& Success probability of exact, correction-free teleportation
& Quantum communication / memory cost reduction \\ \hline

Main resource used
& Multiple copies of the input state
& Representation-theoretic encoding / compression \\ \hline

Operational regime
& Exact probabilistic teleportation in a no-correction setting
& Distributed or programmable quantum information processing with probabilistic resource compression \\ \hline

Receiver-side correction
& Not required
& Not the central resource issue; the protocol is based on coherent matching and heralded postselection \\ \hline

Typical trade-off
& Higher exact-transfer success probability for fixed multicopy resources
& Lower communication / memory cost at the price of reduced success probability \\ \hline

Asymptotic feature
& $p(d,k)\to 1/d$ for large $k$
& Communication / memory can achieve near-optimal, and in some settings optimal, compression, while success probability typically decreases polynomially with the number of uses \\ \hline
\end{tabular}
\caption{Comparison between our multicopy exact teleportation protocol and the
representation-matching framework of Yang and Hayashi~\cite{PRXQuantum.2.020327}.
The two approaches address different operational tasks and optimise different
resources, so they should be viewed as complementary rather than competing.}
\label{tab:comparison_rm}
\end{table}

One future direction is to consider the multicopy state teleportation in a scenario where Alice and Bob share $N$ maximally entangled qudit states, in a multicopy PBT scenario. We consider this generalisation to be of great importance, since it would allow us to increase the success probability by making use of extra entanglement, and it would allow us to generalise the results of learning a unitary operation with $N$ calls when $k$ copies of the input state are allowed.
We believe that this problem is likely to be mathematically challenging, since one should then consider not only symmetries of the $k$ identical states $\ket{\psi}$, but also the symmetries of the $N$ ports.

Second potential research direction is a natural hybrid approach suggested by comparing to the representation-matching protocol. Such hybrid protocols may interpolate between different resource regimes: on the one hand, maximising the success probability at fixed entanglement of program resources, as in our setting; on the other hand, minimising classical communication or memory overhead, as in the representation-matching setting. This points to a concrete future direction, namely the characterization of Pareto-optimal trade-offs among success probability, classical communication, memory, entanglement, number of copies, and finally computational overhead.

Lastly, as aforementioned, our protocol requires a single bit of communication, and, if the protocol is repeated several times,  less than a single bit on average are required. This is a feature of how probabilistic multicopy teleportation is designed. One interesting research direction is to consider teleportation protocols where the classical communication form Alice to Bob is restricted to a single bit, or to less than two dits. In this way, even when corrections are allowed, Alice and Bob cannot perform perfect teleportation by consuming a pair of maximally entangled qudit states. We can then see our construction as a possible candidate for optimal teleportation protocols with a single bit of communication, even in scenarios where corrections are allowed.

\section*{Acknowledgments}
We acknowledge the  Japanese-French Laboratory for Informatics (JFLI) for the support on organising the Japanese-French Quantum Information 2023 workshop. FG and MTQ  acknowledge the PEPR integrated project EPiQ ANR-22-PETQ-0007 part of Plan France 2030 and QIA, which has received funding from the European Union’s Horizon 2020 research and innovation programme under grant agreement No 820445 and from the Horizon Europe grant agreements 101080128 and 101102140. MTQ acknowledges the support of Tremplins nouveaux entrants \& nouvelles entrantes - Edition 2024. MTQ is supported by the Agence Nationale de la Recherche (ANR) through the JCJC programme under grant number ANR-25-CE47-6396-01-HOQO-KS.
MH and TM acknowledge the support of the grant Sonata 16, UMO-2020/39/D/ST2/01234 from the Polish National Science Centre. MS acknowledges support by the IRA Programme, project no. FENG.02.01-IP.05-0006/23, financed by the
FENG program 2021-2027, Priority FENG.02, Measure FENG.02.01., with the support of the FNP.   M.M. acknowledges support by MEXT Quantum Leap Flagship Program (MEXT QLEAP) JPMXS0118069605, JPMXS0120351339, JSPS KAKENHI Grant Number 21H03394 and 23K21643, IBM Quantum, JST CREST Grant Number JPMJCR25I5, and JST ASPIRE Grant Number JPMJAP25A3.
SY acknowledges support by Japan Society for the Promotion of Science (JSPS) KAKENHI Grant Number 23KJ0734, FoPM, WINGS Program, the University of Tokyo, and DAIKIN Fellowship Program, the University of Tokyo.
For the purpose of Open Access, the authors have applied a CC-BY public copyright licence to any Author Accepted Manuscript (AAM) version arising from this submission.

\appendix

\section{Relation to probabilistic simulation of quantum channel from the future to the past} \label{appendix}
In Sec.~V.B of Ref.~\cite{genkina2012optimal}, the authors discuss the optimal probabilistic simulation of quantum channels from $k$ copies of the pure state in the future to a single copy in the past.
This is formulated as a probabilistic quantum channel $\mathcal{C}$ implementable in the form of Fig.~\ref{fig:PBT_many_copies}, which satisfies
\begin{align}
    \label{eq:simulation_future_past}
    \mathcal{C}(P_{1\ldots k}^{sym} \rho P_{1\ldots k}^{sym}) = p \tr_{1\ldots k-1}(P_{1\ldots k}^{sym} \rho P_{1\ldots k}^{sym}) \quad \forall \rho\in \mathcal{L}((\mathbb{C}^d)^{\otimes k}),
\end{align}
with the success probability $p$ independent of the input state $\rho$.
The optimal success probability is shown to be $p = \frac{k}{d(k-1+d)}$ in Ref.~\cite{genkina2012optimal}, which coincides with our result in Theorem~\ref{thm:1}.
However, this does not imply that Theorem~\ref{thm:1} is a straightforward consequence of the result in Ref.~\cite{genkina2012optimal}, as discussed below.
In our formulation, we assume that the input state is given in the form of $\rho = \ketbra{\psi}^{\otimes k}$, i.e., we assume
\begin{align}
    \label{eq:our_setting}
    \mathcal{C}(\ketbra{\psi}^{\otimes k}) = p \ketbra{\psi} \quad \forall \ket{\psi}\in \mathbb{C}^{d} \; \mathrm{s.t.} \; \|\ket{\psi}\|^2 = 1,
\end{align}
while Ref.~\cite{genkina2012optimal} considers arbitrary $\rho$ supported on the symmetric subspace.
Since
\begin{align}
    \mathrm{span}\{\ketbra{\psi}^{\otimes k} \mid \ket{\psi}\in \mathbb{C}^d\} \subsetneq \mathrm{span}\{P_{1\ldots k}^{sym} \rho P_{1\ldots k}^{sym} \mid \rho\in \mathcal{L}((\mathbb{C}^d)^{\otimes k})\}
\end{align}
holds\footnote{The inclusion is strict because the dimension of the left-hand side is given by $\binom{d+2k-1}{2k}$, which is stricly smaller than that of the right-hand side given by $\binom{d+k-1}{k}^2$.}, Eq.~\eqref{eq:simulation_future_past} implies Eq.~\eqref{eq:our_setting} from the linearity of $\mathcal{C}$, but the converse is not shown from the linearity.
Therefore, the achievability of Theorem~\ref{thm:1} can be considered as a corollary of Ref.~\cite{genkina2012optimal}, but the optimality is not trivial.

On the other hand, we can show the implication of Eq.~\eqref{eq:simulation_future_past} from Eq.~\eqref{eq:our_setting} by using the twirling technique, as shown below.
\begin{lemma}
    \label{lem:cp_map_reduction}
If a CP map $\mathcal{C}$ satisfies Eq.~\eqref{eq:our_setting}, then the $\mathrm{SU}(d)$-twirled channel defined by
\begin{align}
    \mathcal{C}'(\cdot) \coloneqq \int_{\mathrm{SU}(d)} \dd U U^\dagger \mathcal{C}(U^{\otimes k}(\cdot)U^{\dagger \otimes k}) U
\end{align}
satisfies Eq.~\eqref{eq:simulation_future_past} with the same success probability $p$.
\end{lemma}
\begin{proof}
Since $\mathcal{C}$ is a CP map, it has a Kraus representation $\mathcal{C}(\cdot) = \sum_j K_j (\cdot) K_j^\dagger$ with $K_j:(\mathbb{C}^d)^{\otimes k} \to \mathbb{C}^d$.
Then, Eq.~\eqref{eq:our_setting} reads
\begin{align}
    \sum_j K_j \ketbra{\psi}^{\otimes k} K_j^\dagger = p\ketbra{\psi},
\end{align}
i.e., there exists $\alpha_j^{\psi} \in \mathbb{C}$ such that
\begin{align}
    K_j \ket{\psi}^{\otimes k} &= \alpha_j^{\psi} \ket{\psi},\\
    \sum_j \abs{\alpha_j^\psi}^2 &= p \quad \forall \ket{\psi}\in \mathbb{C}^{d} \; \mathrm{s.t.} \; \|\ket{\psi}\|^2 = 1.
\end{align}
By linearity of $K_j$, $\alpha_j^{\psi}$ is given as a $(k-1)$-th polynomial of $\ket{\psi}$, i.e., there exists a vector $\ket{\eta_j}\in (\mathbb{C}^{d})^{\otimes (k-1)}$ such that
\begin{align}
    \alpha_j^\psi = \bra{\eta_j} \ket{\psi}^{\otimes (k-1)},
\end{align}
and $\ket{\eta_j}$ is in the symmetric subspace of $(\mathbb{C}^d)^{\otimes (k-1)}$.
Therefore, $K_j$ satisfies
\begin{align}
    K_j \ket{\psi}^{\otimes k} = (\bra{\eta_j} \otimes \1_d) \ket{\psi}^{\otimes k} \quad \forall \ket{\psi}\in \mathbb{C}^d,
\end{align}
i.e.,
\begin{align}
    K_j P_{1\ldots k}^{sym} = (\bra{\eta_j} \otimes \1_d) P_{1\ldots k}^{sym}
\end{align}
holds.
Then, the action of $\mathcal{C}'$ on $P_{1\ldots k}^{sym} \rho P_{1\ldots k}^{sym}$ is given by
\begin{align}
    \mathcal{C}'(P_{1\ldots k}^{sym} \rho P_{1\ldots k}^{sym})
    &= \sum_j \int_{\mathrm{SU}(d)} \dd U U^\dagger K_j U^{\otimes k}(P_{1\ldots k}^{sym} \rho P_{1\ldots k}^{sym})U^{\dagger \otimes k} K_j^\dagger U\\
    &= \sum_j \int_{\mathrm{SU}(d)} \dd U U^\dagger K_j P_{1\ldots k}^{sym} U^{\otimes k} \rho U^{\dagger \otimes k} P_{1\ldots k}^{sym} K_j^\dagger U\\
    &= \sum_j \int_{\mathrm{SU}(d)} \dd U (\bra{\eta_j} \otimes U^\dagger) P_{1\ldots k}^{sym} U^{\otimes k} \rho U^{\dagger \otimes k} P_{1\ldots k}^{sym} (\ket{\eta_j} \otimes U)\\
    &= \sum_j \int_{\mathrm{SU}(d)} \dd U (\bra{\eta_j} U^{\otimes (k-1)} \otimes \1_d) P_{1\ldots k}^{sym}\rho P_{1\ldots k}^{sym} (U^{\dagger \otimes (k-1)} \ket{\eta_j} \otimes \1_d)\\
    &= \tr_{1\ldots k-1} (P_{1\ldots k}^{sym} \rho P_{1\ldots k}^{sym} (K\otimes \1_d)),
\end{align}
where $K$ is defined by
\begin{align}
    K\coloneqq \sum_j \int_{\mathrm{SU}(d)} \dd U (U^{\dagger \otimes (k-1)} \ketbra{\eta_j} U^{\otimes (k-1)}).
\end{align}
Due to the left- and right-invariance of the Haar measure, $K$ commutes with $U^{\otimes (k-1)}$ for every $U\in \mathrm{SU}(d)$.
Since $\ket{\eta_j}$ is in the symmetric subspace of $(\mathbb{C}^d)^{\otimes (k-1)}$, Schur's lemma implies that $K$ is proportional to the identity operator on the symmetric subspace, i.e.,
\begin{align}
    K = c P_{1\ldots k-1}^{sym}
\end{align}
holds for some $c>0$.
Therefore, we obtain
\begin{align}
    \mathcal{C}'(P_{1\ldots k}^{sym} \rho P_{1\ldots k}^{sym}) &= c \tr_{1\ldots k-1}(P_{1\ldots k}^{sym} \rho P_{1\ldots k}^{sym}).
\end{align}
On the other hand, Eq.~\eqref{eq:our_setting} implies that
\begin{align}
    \mathcal{C}'(\ketbra{\psi}^{\otimes k}) &= p \ketbra{\psi} \quad \forall \ket{\psi}\in \mathbb{C}^{d} \; \mathrm{s.t.} \; \|\ket{\psi}\|^2 = 1.
\end{align}
By comparing the two equations, we conclude that $c=p$, which completes the proof.
\end{proof}
By using Lemma~\ref{lem:cp_map_reduction} and Ref.~\cite{genkina2012optimal}, we can show the optimality of Theorem~\ref{thm:1} as follows.
\begin{proof}[Alternative proof of the optimality part of Theorem~\ref{thm:1}]
We show the optimality by contradiction.
Suppose there exists a CP map $\mathcal{C}$ implementable in the setting of Fig.~\ref{fig:PBT_many_copies}, which satisfies Eq.~\eqref{eq:our_setting} with success probability $p > \frac{k}{d(k-1+d)}$.
Then, by Lemma~\ref{lem:cp_map_reduction}, the $\mathrm{SU}(d)$-twirled map $\mathcal{C}'$ satisfies Eq.~\eqref{eq:simulation_future_past} with the same success probability $p$.
Since the CP map $\mathcal{C}'$ is also implementable in the setting of Fig.~\ref{fig:PBT_many_copies}, this implies that there exists a probabilistic simulation of quantum channels from $k$ copies of the pure state in the future to a single copy in the past with the success probability $p > \frac{k}{d(k-1+d)}$, which contradicts with the optimality result in Ref.~\cite{genkina2012optimal}.
\end{proof}


\nocite{apsrev42Control} 
\bibliographystyle{0_MTQ_apsrev4-2_corrected}
\bibliography{0_MTQ_bib.bib}
\end{document}

%% file: 0_MTQ_macros.tex
\newcommand{\1}{\mbox{1}\hspace{-0.25em}\mbox{l}}
\newcommand{\id}{\mathds{1}}

\renewcommand{\H}{\mathcal{H}}

\renewcommand{\L}{\mathcal{L}}

\newcommand{\SU}{\mathcal{SU}}

\newcommand{\<}{\langle}
\renewcommand{\>}{\rangle}

\makeatletter 
\newcommand{\doublewidetilde}[1]{{%
  \mathpalette\double@widetilde{#1}%
}}
\newcommand{\double@widetilde}[2]{%
  \sbox\z@{$\m@th#1\widetilde{#2}$}%
  \ht\z@=.9\ht\z@
  \widetilde{\box\z@}%
}
  
\makeatother

\newtheorem{definition}{Definition}

\newtheorem{lemma}{Lemma}



%% file: 0_MTQ_bib.bib
@CONTROL{apsrev42Control,author="",editor="1",pages="1",title="0",year="0"}

@article{Schumacher1995coding,
  title = {Quantum coding},
  author = {Schumacher, Benjamin},
  journal = {Phys. Rev. A},
  volume = {51},
  issue = {4},
  pages = {2738--2747},
  numpages = {0},
  year = {1995},
  month = {Apr},
  publisher = {American Physical Society},
  doi = {10.1103/PhysRevA.51.2738},
  url = {https://link.aps.org/doi/10.1103/PhysRevA.51.2738}
}

@ARTICLE{Yang2016compression,
       author = {{Yang}, Yuxiang and {Chiribella}, Giulio and {Ebler}, Daniel},
        title = "{Efficient Quantum Compression for Ensembles of Identically Prepared Mixed States}",
       journal = {Phys. Rev. Lett.},
         year = 2016,
        month = feb,
       volume = {116},
       number = {8},
          eid = {080501},
        pages = {080501},
          doi = {10.1103/PhysRevLett.116.080501},
archivePrefix = {arXiv},
       eprint = {1506.03542},
 primaryClass = {quant-ph},
       adsurl = {https://ui.adsabs.harvard.edu/abs/2016PhRvL.116h0501Y}
}

@ARTICLE{Beigi2011Simplified,
       author = {{Beigi}, Salman and {K{\"o}nig}, Robert},
        title = "{Simplified instantaneous non-local quantum computation with applications to position-based cryptography}",
      journal = {New Journal of Physics},
         year = 2011,
        month = sep,
       volume = {13},
       number = {9},
          eid = {093036},
        pages = {093036},
          doi = {10.1088/1367-2630/13/9/093036},
archivePrefix = {arXiv},
       eprint = {1101.1065},
 primaryClass = {quant-ph},
       adsurl = {https://ui.adsabs.harvard.edu/abs/2011NJPh...13i3036B}
}

@ARTICLE{Chitambar2023Duality,
  author={Chitambar, Eric and Leditzky, Felix},
  journal={IEEE Transactions on Information Theory}, 
  title={On the Duality of Teleportation and Dense Coding}, 
  year={2024},
  volume={70},
  number={5},
  pages={3529-3537},
  doi={10.1109/TIT.2023.3331821},
archivePrefix = {arXiv},
       eprint = {2302.14798},
 primaryClass = {quant-ph}
}

@ARTICLE{Briegel1998Repeaters,
  title = {Quantum Repeaters: The Role of Imperfect Local Operations in Quantum Communication},
  author = {Briegel, H.-J. and D\"ur, W. and Cirac, J. I. and Zoller, P.},
  journal = {Phys. Rev. Lett.},
  volume = {81},
  issue = {26},
  pages = {5932--5935},
  year = {1998},
  month = {Dec},
  publisher = {American Physical Society},
  doi = {10.1103/PhysRevLett.81.5932},
archivePrefix = {arXiv},
       eprint = {quant-ph/9803056},
 primaryClass = {quant-ph}
}

@ARTICLE{Bennett1996Mixed,
       author = {{Bennett}, Charles H. and {Divincenzo}, David P. and {Smolin}, John A. and {Wootters}, William K.},
        title = "{Mixed-state entanglement and quantum error correction}",
      journal = {Phys. Rev.~A},
         year = 1996,
        month = nov,
       volume = {54},
       number = {5},
        pages = {3824-3851},
          doi = {10.1103/PhysRevA.54.3824},
archivePrefix = {arXiv},
       eprint = {quant-ph/9604024},
 primaryClass = {quant-ph},
       adsurl = {https://ui.adsabs.harvard.edu/abs/1996PhRvA..54.3824B}
}

@ARTICLE{Gottesman1999Demonstrating,
       author = {{Gottesman}, Daniel and {Chuang}, Isaac L.},
        title = "{Demonstrating the viability of universal quantum computation using teleportation and single-qubit operations}",
      journal = {Nature},
         year = 1999,
        month = nov,
       volume = {402},
       number = {6760},
        pages = {390-393},
          doi = {10.1038/46503},
archivePrefix = {arXiv},
       eprint = {quant-ph/9908010},
 primaryClass = {quant-ph},
       adsurl = {https://ui.adsabs.harvard.edu/abs/1999Natur.402..390G}
}

@ARTICLE{Brizc2024PBRSP,
       author = {{Brzi{\'c}}, Vanessa and {Yoshida}, Satoshi and {Murao}, Mio and { Quintino}, Marco T.} ,
        title = "{Higher-order quantum computing with known input states}",
      journal = {arXiv e-prints},
         year = 2025,
        month = oct,
archivePrefix = {arXiv},
       eprint = {2510.20530},
 primaryClass = {quant-ph},
       adsurl = {https://ui.adsabs.harvard.edu/abs/2025arXiv251020530B}
}

@ARTICLE{Pirandola2020crypto,
       author = {{Pirandola}, S. and {Andersen}, U.~L. and {Banchi}, L. and {Berta}, M. and {Bunandar}, D. and {Colbeck}, R. and {Englund}, D. and {Gehring}, T. and {Lupo}, C. and {Ottaviani}, C. and {Pereira}, J.~L. and {Razavi}, M. and {Shamsul Shaari}, J. and {Tomamichel}, M. and {Usenko}, V.~C. and {Vallone}, G. and {Villoresi}, P. and {Wallden}, P.},
        title = "{Advances in quantum cryptography}",
      journal = {Advances in Optics and Photonics},
         year = 2020,
        month = dec,
       volume = {12},
       number = {4},
        pages = {1012},
          doi = {10.1364/AOP.361502},
archivePrefix = {arXiv},
       eprint = {1906.01645},
 primaryClass = {quant-ph},
       adsurl = {https://ui.adsabs.harvard.edu/abs/2020AdOP...12.1012P}
}

@ARTICLE{Muguruza2024PBRSP,
       author = {{Muguruza}, Garazi and {Speelman}, Florian},
        title = "{Port-Based State Preparation and Applications}",
      journal = {arXiv e-prints},
         year = 2024,
        month = feb,
archivePrefix = {arXiv},
       eprint = {2402.18356},
 primaryClass = {quant-ph},
       adsurl = {https://ui.adsabs.harvard.edu/abs/2024arXiv240218356M}
}

@ARTICLE{Quintino2021Deterministic,
archivePrefix = {arXiv},
       eprint = {2109.08202},
 primaryClass = {quant-ph},
   title={Deterministic transformations between unitary operations: Exponential advantage with adaptive quantum circuits and the power of indefinite causality},
   volume={6},
   ISSN={2521-327X},
   url={http://dx.doi.org/10.22331/q-2022-03-31-679},
   DOI={10.22331/q-2022-03-31-679},
   journal={Quantum},
   publisher={Verein zur Forderung des Open Access Publizierens in den Quantenwissenschaften},
   author={Quintino, Marco Túlio and Ebler, Daniel},
   year={2022},
   month=mar, pages={679} }

@book{WatrousBook,
place={Cambridge},
title={The Theory of Quantum Information},
publisher={Cambridge University Press},
author={Watrous, John},
year={2018}}

@ARTICLE{Harrow2013Church,
       author = {{Harrow}, Aram W.},
        title = "{The Church of the Symmetric Subspace}",
      journal = {arXiv e-prints},
         year = 2013,
        month = aug,
archivePrefix = {arXiv},
       eprint = {1308.6595},
 primaryClass = {quant-ph},
       adsurl = {https://ui.adsabs.harvard.edu/abs/2013arXiv1308.6595H}
}

@ARTICLE{Yoshida2024one,
       author = {{Yoshida}, Satoshi and {Koizumi}, Yuki and {Studzi{\'n}ski}, Micha{\l} and {Quintino}, Marco T. and {Murao}, Mio},
        title = "{One-to-one Correspondence between Deterministic Port-Based Teleportation and Unitary Estimation}",
volume={72},
   ISSN={1557-9654},
   DOI={10.1109/tit.2026.3658543},
   number={4},
   journal={IEEE Transactions on Information Theory},
   publisher={Institute of Electrical and Electronics Engineers (IEEE)},
   year={2026},
   month=apr, pages={2358–2377},
archivePrefix = {arXiv},
       eprint = {2408.11902},
 primaryClass = {quant-ph},
       adsurl = {https://ui.adsabs.harvard.edu/abs/2024arXiv240811902Y}
}

@ARTICLE{nielsen1997programmable,
       author = {{Nielsen}, M.~A. and {Chuang}, Isaac L.},
        title = "{Programmable Quantum Gate Arrays}",
      journal = {Phys. Rev. Lett.},
         year = 1997,
        month = jul,
       volume = {79},
       number = {2},
        pages = {321-324},
          doi = {10.1103/PhysRevLett.79.321},
archivePrefix = {arXiv},
       eprint = {quant-ph/9703032},
 primaryClass = {quant-ph},
       adsurl = {https://ui.adsabs.harvard.edu/abs/1997PhRvL..79..321N}
}

@misc{fei2023,
      title={Efficient Quantum Algorithm for Port-based Teleportation}, 
      author={Jiani Fei and Sydney Timmerman and Patrick Hayden},
      year={2023},
      eprint={2310.01637},
      archivePrefix={arXiv},
      primaryClass={quant-ph},
      url={https://arxiv.org/abs/2310.01637}, 
}

@article{grinko2024,
       author = {{Grinko}, Dmitry and {Burchardt}, Adam and {Ozols}, Maris},
        title = "{Efficient quantum circuits for port-based teleportation}",
      journal = {arXiv e-prints},
         year = 2023,
        month = dec,
archivePrefix = {arXiv},
       eprint = {2312.03188},
 primaryClass = {quant-ph},
       adsurl = {https://ui.adsabs.harvard.edu/abs/2023arXiv231203188G}
}

@article{grinko2023,
       author = {{Grinko}, Dmitry and {Ozols}, Maris},
        title = "{Linear programming with unitary-equivariant constraints}",
      journal = {arXiv e-prints},
         year = 2022,
        month = jul,
archivePrefix = {arXiv},
       eprint = {2207.05713},
 primaryClass = {quant-ph},
       adsurl = {https://ui.adsabs.harvard.edu/abs/2022arXiv220705713G}
}

@book{ceccherini2010representation,
  title={{Representation Theory of the Symmetric Groups: The Okounkov-Vershik Approach, Character Formulas, and Partition Algebras}},
  author={Ceccherini-Silberstein, Tullio and Scarabotti, Fabio and Tolli, Filippo},
  volume={121},
  year={2010},
  publisher={Cambridge University Press},
  doi={10.1017/CBO9781139192361}
}

@book{fulton1997young,
  title={{Young Tableaux: With Applications to Representation Theory and Geometry}},
  author={Fulton, William},
  number={35},
  year={1997},
  publisher={Cambridge University Press},
  doi={10.1017/CBO9780511626241}
}

@article{Mozrzymas_2018JPA,
doi = {10.1088/1751-8121/aaad15},
url = {https://dx.doi.org/10.1088/1751-8121/aaad15},
year = {2018},
month = {feb},
publisher = {IOP Publishing},
volume = {51},
number = {12},
pages = {125202},
author = {Marek Mozrzymas and Michał Studziński and Michał Horodecki},
title = {A simplified formalism of the algebra of partially transposed permutation operators with applications},
journal = {Journal of Physics A: Mathematical and Theoretical}
}

@ARTICLE{mozrzymas18,
       author = {{Mozrzymas}, Marek and {Studzi{\'n}ski}, Micha{\l} and {Strelchuk}, Sergii and {Horodecki}, Micha{\l}},
        title = "{Optimal port-based teleportation}",
      journal = {New Journal of Physics},
         year = 2018,
        month = may,
       volume = {20},
       number = {5},
          eid = {053006},
        pages = {053006},
          doi = {10.1088/1367-2630/aab8e7},
archivePrefix = {arXiv},
       eprint = {1707.08456},
 primaryClass = {quant-ph},
       adsurl = {https://ui.adsabs.harvard.edu/abs/2018NJPh...20e3006M}
}

@ARTICLE{Bennett01RSP,
       author = {{Bennett}, Charles H. and {Divincenzo}, David P. and {Shor}, Peter W. and {Smolin}, John A. and {Terhal}, Barbara M. and {Wootters}, William K.},
        title = "{Remote State Preparation}",
      journal = {Phys. Rev. Lett.},
         year = 2001,
        month = aug,
       volume = {87},
       number = {7},
          eid = {077902},
        pages = {077902},
          doi = {10.1103/PhysRevLett.87.077902},
archivePrefix = {arXiv},
       eprint = {quant-ph/0006044},
 primaryClass = {quant-ph},
       adsurl = {https://ui.adsabs.harvard.edu/abs/2001PhRvL..87g7902B}
}

@ARTICLE{Jozsa05GateTeleportation,
       author = {{Jozsa}, Richard},
        title = "{An introduction to measurement based quantum computation}",
      journal = {arXiv e-prints},
         year = 2005,
        month = aug,
       eprint = {quant-ph/0508124},
 primaryClass = {quant-ph},
       adsurl = {https://ui.adsabs.harvard.edu/abs/2005quant.ph..8124J}
}

@ARTICLE{buhrman16BellPBT,
       author = {{Buhrman}, Harry and {Czekaj}, {L}ukasz and {Grudka}, Andrzej and {Horodecki}, Micha{\l} and {Horodecki}, Pawe{\l} and {Markiewicz}, Marcin and {Speelman}, Florian and {Strelchuk}, Sergii},
        title = "{Quantum communication complexity advantage implies violation of a Bell inequality}",
      journal = {Proceedings of the National Academy of Science},
         year = 2016,
        month = mar,
       volume = {113},
       number = {12},
        pages = {3191-3196},
          doi = {10.1073/pnas.1507647113},
archivePrefix = {arXiv},
       eprint = {1502.01058},
 primaryClass = {quant-ph},
       adsurl = {https://ui.adsabs.harvard.edu/abs/2016PNAS..113.3191B}
}

@ARTICLE{ishizaka09,
       author = {{Ishizaka}, Satoshi and {Hiroshima}, Tohya},
        title = "{Quantum teleportation scheme by selecting one of multiple output ports}",
      journal = {Phys. Rev.~A},
         year = 2009,
        month = apr,
       volume = {79},
       number = {4},
          eid = {042306},
        pages = {042306},
          doi = {10.1103/PhysRevA.79.042306},
archivePrefix = {arXiv},
       eprint = {0901.2975},
 primaryClass = {quant-ph},
       adsurl = {https://ui.adsabs.harvard.edu/abs/2009PhRvA..79d2306I}
}

@ARTICLE{studzinski20,
  author={Studziński, Michał and Mozrzymas, Marek and Kopszak, Piotr and Horodecki, Michał},
  journal={IEEE Transactions on Information Theory}, 
  title={Efficient Multi Port-Based Teleportation Schemes}, 
  year={2022},
  volume={68},
  number={12},
  pages={7892-7912},
  doi={10.1109/TIT.2022.3187852}}

@article{mozrzymas2014,
    author = {Mozrzymas, Marek and Horodecki, Michał and Studziński, Michał},
    title = "{Structure and properties of the algebra of partially transposed permutation operators}",
    journal = {Journal of Mathematical Physics},
    volume = {55},
    number = {3},
    pages = {032202},
    year = {2014},
    month = {03},
    issn = {0022-2488},
    doi = {10.1063/1.4869027}
}

@ARTICLE{murao99,
       author = {{Murao}, M. and {Jonathan}, D. and {Plenio}, M.~B. and {Vedral}, V.},
        title = "{Quantum telecloning and multiparticle entanglement}",
      journal = {Phys. Rev.~A},
         year = 1999,
        month = jan,
       volume = {59},
       number = {1},
        pages = {156-161},
          doi = {10.1103/PhysRevA.59.156},
archivePrefix = {arXiv},
       eprint = {quant-ph/9806082},
 primaryClass = {quant-ph},
       adsurl = {https://ui.adsabs.harvard.edu/abs/1999PhRvA..59..156M}
}

@ARTICLE{quintino19PRA,
       author = {{Quintino}, Marco T{\'u}lio and {Dong}, Qingxiuxiong and {Shimbo}, Atsushi and {Soeda}, Akihito and {Murao}, Mio},
        title = "{Probabilistic exact universal quantum circuits for transforming unitary operations}",
      journal =  {Phys. Rev.~A},
         year = 2019,
        month = dec,
       volume = {100},
       number = {6},
          eid = {062339},
        pages = {062339},
          doi = {10.1103/PhysRevA.100.062339},
archivePrefix = {arXiv},
       eprint = {1909.01366},
 primaryClass = {quant-ph},
       adsurl = {https://ui.adsabs.harvard.edu/abs/2019PhRvA.100f2339Q}
}

@ARTICLE{quintino19PRL,
       author = {{Quintino}, Marco T{\'u}lio and {Dong}, Qingxiuxiong and {Shimbo}, Atsushi and {Soeda}, Akihito and {Murao}, Mio},
        title = "{Reversing Unknown Quantum Transformations: Universal Quantum Circuit for Inverting General Unitary Operations}",
      journal = {{Phys. Rev. Lett.},},
         year = 2019,
        month = nov,
       volume = {123},
       number = {21},
          eid = {210502},
        pages = {210502},
          doi = {10.1103/PhysRevLett.123.210502},
archivePrefix = {arXiv},
       eprint = {1810.06944},
 primaryClass = {quant-ph},
       adsurl = {https://ui.adsabs.harvard.edu/abs/2019PhRvL.123u0502Q}
}

@ARTICLE{sedlak18SAR,
       author = {{Sedl{\'a}k}, Michal and {Bisio}, Alessandro and {Ziman}, M{\'a}rio},
        title = "{Optimal Probabilistic Storage and Retrieval of Unitary Channels}",
      journal = {Phys. Rev. Lett.},
         year = "2019",
        month = "May",
       volume = {122},
       number = {17},
          eid = {170502},
        pages = {170502},
          doi = {10.1103/PhysRevLett.122.170502},
archivePrefix = {arXiv},
       eprint = {1809.04552},
 primaryClass = {quant-ph},
       adsurl = {https://ui.adsabs.harvard.edu/abs/2019PhRvL.122q0502S}
}

@ARTICLE{bisio10Learning,
   author = {{Bisio}, A. and {Chiribella}, G. and {D'Ariano}, G.~M. and {Facchini}, S. and 
	{Perinotti}, P.},
    title = "{Optimal quantum learning of a unitary transformation}",
  journal = {Phys. Rev.~A},
archivePrefix = "arXiv",
   eprint = {0903.0543},
 primaryClass = "quant-ph",
     year = 2010,
    month = mar,
   volume = 81,
   number = 3,
      eid = {032324},
    pages = {032324},
      doi = {10.1103/PhysRevA.81.032324},
   adsurl = {http://adsabs.harvard.edu/abs/2010PhRvA..81c2324B}
}

@ARTICLE{Gottesman99GateTeleportation,
   author = {{Gottesman}, D. and {Chuang}, I.~L.},
    title = "{Demonstrating the viability of universal quantum computation using teleportation and single-qubit operations}",
  journal = {Nature},
   eprint = {quant-ph/9908010},
     year = 1999,
    month = nov,
   volume = 402,
    pages = {390-393},
      doi = {10.1038/46503},
   adsurl = {http://adsabs.harvard.edu/abs/1999Natur.402..390G}
}

@ARTICLE{ishizaka08,
   author = {{Ishizaka}, S. and {Hiroshima}, T.},
    title = "{Asymptotic Teleportation Scheme as a Universal Programmable Quantum Processor}",
  journal = {Phys. Rev. Lett.},
archivePrefix = "arXiv",
   eprint = {0807.4568},
 primaryClass = "quant-ph",
     year = 2008,
    month = dec,
   volume = 101,
   number = 24,
      eid = {240501},
    pages = {240501},
      doi = {10.1103/PhysRevLett.101.240501},
   adsurl = {http://adsabs.harvard.edu/abs/2008PhRvL.101x0501I}
}

@ARTICLE{studzinski16,
   author = {{Studzi{\'n}ski}, M. and {Strelchuk}, S. and {Mozrzymas}, M. and 
	{Horodecki}, M.},
    title = "{Port-based teleportation in arbitrary dimension}",
  journal = {Scientific Reports},
archivePrefix = "arXiv",
   eprint = {1612.09260},
 primaryClass = "quant-ph",
     year = 2017,
    month = sep,
   volume = 7,
      eid = {10871},
    pages = {10871},
      doi = {10.1038/s41598-017-10051-4},
   adsurl = {http://adsabs.harvard.edu/abs/2017NatSR...710871S}
}

@article{bennett93,
  title = {Teleporting an unknown quantum state via dual classical and Einstein-Podolsky-Rosen channels},
  author = {Bennett, Charles H. and Brassard, Gilles and Cr\'epeau, Claude and Jozsa, Richard and Peres, Asher and Wootters, William K.},
  journal = {Phys. Rev. Lett.},
  volume = {70},
  issue = {13},
  pages = {1895--1899},
  numpages = {0},
  year = {1993},
  month = {Mar},
  publisher = {American Physical Society},
  doi = {10.1103/PhysRevLett.70.1895},
  url = {https://link.aps.org/doi/10.1103/PhysRevLett.70.1895}
}

@ARTICLE{strelchuk13,
       author = {{Strelchuk}, Sergii and {Horodecki}, Micha{\l} and {Oppenheim}, Jonathan},
        title = "{Generalized Teleportation and Entanglement Recycling}",
  journal = {Phys. Rev. Lett.},
         year = 2013,
        month = jan,
       volume = {110},
       number = {1},
          eid = {010505},
        pages = {010505},
          doi = {10.1103/PhysRevLett.110.010505},
archivePrefix = {arXiv},
       eprint = {1209.2683},
 primaryClass = {quant-ph},
       adsurl = {https://ui.adsabs.harvard.edu/abs/2013PhRvL.110a0505S}
}

@book{NielsenChuangBook,
  title={Quantum Computation and Quantum Information},
  author={Nielsen, M.A. and Chuang, I.L.},
  isbn={9780521635035},
  lccn={98022029},
  series={Cambridge Series on Information and the Natural Sciences},
  year={2000},
  publisher={Cambridge University Press}
}

@ARTICLE{horodecki_review,
   author = {{Horodecki}, R. and {Horodecki}, P. and {Horodecki}, M. and 
	{Horodecki}, K.},
    title = "{Quantum entanglement}",
  journal = {Reviews of Modern Physics},
   eprint = {quant-ph/0702225},
     year = 2009,
    month = apr,
   volume = 81,
    pages = {865-942},
      doi = {10.1103/RevModPhys.81.865},
   adsurl = {http://adsabs.harvard.edu/abs/2009RvMP...81..865H}
}

@ARTICLE{vershik,
  author = {{A. M. Vershik}},
  title = {{Hooks formula and related identities}},
  journal = {J. Soviet Math., 59 (1992), 1029–1040},
  year = {1992},
  volume = {59:5},
  pages = {1029–1040},
  issue = {5},
  doi = {10.1007/BF01480684}
}

@article{wildeShannon, 
       author = {{Wilde}, Mark M.},
        title = "{From Classical to Quantum Shannon Theory}",
      journal = {arXiv e-prints},
         year = 2011,
        month = jun,
archivePrefix = {arXiv},
       eprint = {1106.1445},
 primaryClass = {quant-ph}
}

@article{boschi98XPteleport,
       author = {{Boschi}, D. and {Branca}, S. and {de Martini}, F. and {Hardy}, L. and {Popescu}, S.},
        title = "{Experimental Realization of Teleporting an Unknown Pure Quantum State via Dual Classical and Einstein-Podolsky-Rosen Channels}",
  journal = {Phys. Rev. Lett. },
         year = 1998,
        month = feb,
       volume = {80},
       number = {6},
        pages = {1121-1125},
          doi = {10.1103/PhysRevLett.80.1121},
archivePrefix = {arXiv},
       eprint = {quant-ph/9710013},
 primaryClass = {quant-ph},
       adsurl = {https://ui.adsabs.harvard.edu/abs/1998PhRvL..80.1121B}
}

@article{grinko2023gelfand,
       author = {{Grinko}, Dmitry and {Burchardt}, Adam and {Ozols}, Maris},
        title = "{Gelfand-Tsetlin basis for partially transposed permutations, with applications to quantum information}",
      journal = {arXiv e-prints},
         year = 2023,
        month = oct,
archivePrefix = {arXiv},
       eprint = {2310.02252},
 primaryClass = {quant-ph}
}

@article{genkina2012optimal,
  title = {Optimal probabilistic simulation of quantum channels from the future to the past},
  author = {Genkina, Dina and Chiribella, Giulio and Hardy, Lucien},
  journal = {Phys. Rev. A},
  volume = {85},
  issue = {2},
  pages = {022330},
  numpages = {14},
  year = {2012},
  month = {Feb},
  publisher = {American Physical Society},
  doi = {10.1103/PhysRevA.85.022330},
  url = {https://link.aps.org/doi/10.1103/PhysRevA.85.022330},
  archivePrefix = {arXiv},
  eprint = {1112.1469},
}

@article{PRXQuantum.2.020327,
  title = {Representation Matching For Remote Quantum Computing},
  author = {Yang, Yuxiang and Hayashi, Masahito},
  journal = {PRX Quantum},
  volume = {2},
  issue = {2},
  pages = {020327},
  numpages = {14},
  year = {2021},
  month = {May},
  publisher = {American Physical Society},
  doi = {10.1103/PRXQuantum.2.020327},
  url = {https://link.aps.org/doi/10.1103/PRXQuantum.2.020327}
}
